\documentclass{article}

\usepackage[a4paper,top=2cm,bottom=2cm,left=3cm,right=3cm,marginparwidth=1.75cm]{geometry}
\usepackage{amsmath}
\usepackage{amssymb}
\usepackage{amsthm}
\usepackage{mathtools}
\usepackage{graphicx}
\usepackage{hyperref}
\usepackage[all]{xy}
\usepackage{enumitem}
\usepackage{authblk}
\usepackage{xcolor}
\usepackage{cleveref}
\usepackage{microtype}
\usepackage{booktabs}

\theoremstyle{plain}
\newtheorem{theorem}{Theorem}[section]
\newtheorem{lemma}[theorem]{Lemma}

\newtheorem{corollary}[theorem]{Corollary}

\theoremstyle{definition}
\newtheorem{definition}[theorem]{Definition}
\newtheorem{example}[theorem]{Example}
\theoremstyle{remark}
\newtheorem{remark}[theorem]{Remark}

\newcommand{\Z}{\mathbb{Z}}
\newcommand{\R}{\mathbb{R}}
\newcommand{\RR}{\mathrm{R}}
\newcommand{\RP}{\mathbb{R}\mathit{P}}
\newcommand{\A}{\mathcal{A}}

\renewcommand{\SS}{\mathcal{S}}
\newcommand{\col}{\mathrm{col}}

\newcommand{\fg}{\mathrm{fg}}
\newcommand{\incl}{\mathrm{incl}}
\newcommand{\conj}{\mathrm{conj}}
\newcommand{\SO}{\mathrm{SO}}
\newcommand{\SU}{\mathrm{SU}}
\newcommand{\BTet}{2T}
\newcommand{\BOct}{2O}

\renewcommand{\Im}{\operatorname{Im}}
\newcommand{\Int}{\operatorname{Int}}
\newcommand{\Map}{\mathrm{Map}}
\newcommand{\Hom}{\mathrm{Hom}}
\newcommand{\ang}[1]{\langle#1\rangle}
\newcommand{\angg}[1]{\langle\!\langle#1\rangle\!\rangle}

\newcommand{\defeq}{\vcentcolon=}

\usepackage[
	backend=biber,
	uniquelist=false,
  	maxsortnames=2,
  	maxcitenames=2,
  	maxbibnames=100,
  	sorting=nyt,
  	url=false,
  	eprint=false
]{biblatex}
\title{Homotopy classification of knotted defects in bounded domains}
\author[1,4]{Yuta Nozaki}
\author[2]{David Palmer}
\author[3,4]{Yuya Koda}

\affil[1]{\footnotesize Faculty of Environment and Information Sciences, Yokohama National University, Yokohama 240-8501, Japan;
\texttt{nozaki-yuta-vn@ynu.ac.jp}}

\affil[2]{\footnotesize School of Engineering and Applied Sciences, Harvard University, Cambridge, Massachusetts 02138, USA;
\texttt{dpalmer@seas.harvard.edu}}

\affil[3]{\footnotesize Department of Mathematics, Hiyoshi Campus, Keio University, Yokohama 223-8521, Japan;
\texttt{koda@keio.jp}}

\affil[4]{\footnotesize International Institute for Sustainability with Knotted Chiral Meta Matter (WPI-SKCM$^2$), Hiroshima University, 1-3-1 Kagamiyama, Higashi-Hiroshima, Hiroshima 739-8526, Japan}

\begin{document}
\date{}
\maketitle

\begin{abstract}
Nozaki et.~al.\ gave a homotopy classification of the knotted defects of ordered media in three-dimensional space by considering continuous maps from complements of spatial graphs to the order parameter space modulo a certain equivalence relation.
We extend their result by giving a classification scheme for ordered media in handlebodies, where defects are allowed to reach the boundary.
Through monodromies around meridional loops, global defects are described in terms of planar diagrams whose edges are colored by elements of the fundamental group of the order parameter space.
We exhibit examples of this classification in octahedral frame fields and biaxial nematic liquid crystals.
\end{abstract}

%%%%
\section{Introduction}
\label{sec:intro}

The relationship between ordered media and their topological defects is a central object of study in condensed matter physics. Though the classification of individual defects via homotopy theory has been understood for decades \cite{PoTo77,Poe81,Mer79,NHM88,VoMi77}, the global classification of defect \emph{configurations} in three dimensions remains a challenge. 
To put it simply, the problem is, given a pair of defect configurations, to determine whether they are equivalent under transformations arising from the physical evolution of the underlying order parameter field. These allowed transformations include, at a minimum, isotopy of the defects.

The most-studied case is that of \emph{nematic} symmetry, as observed in liquid crystals. For three-dimensional nematics, the order parameter space is $\RP^2$, which has fundamental group $\Z_2=\Z/2\Z$. 
For individual line defects, this is perhaps the simplest possible case. 
However, the search for a global classification is made more complicated by the fact that $\pi_2(\RP^2)$ is non-zero \cite{MaAl14,MaAl16}. More-exotic ordered media, such as the hypothetical biaxial nematics \cite{tschierskeBiaxialNematicPhases2010}, feature defect lines classified by elements of a non-abelian group. This renders the global classification problem more difficult---for one thing, individual defect labels are only defined up to conjugation.

Outside of soft-matter physics, non-abelian topological defect configurations arise in the hexahedral meshing problem in geometry processing and computational engineering. This is the problem of decomposing a domain in $\R^3$ into topological cubes satisfying certain combinatorial conditions. 
Following related work in two dimensions, \cite{nieserCubeCoverParameterization3D2011} proposed a method for hexahedral meshing of volumes by first computing an octahedral bundle with defects (known in the meshing literature as singularities). Such a bundle can be encoded by an order parameter field with octahedral symmetry, a so-called \emph{octahedral frame field}~\cite{huangBoundaryAlignedSmooth2011} (see \Cref{fig:octa-examples} and Appendix~\ref{sec: The subgroups of the binary octahedral group}). 
To that end, a line of research in geometry processing has studied discretizing, parametrizing, and optimizing such fields \cite{liAllhexMeshingUsing2012,jiangFrameFieldSingularity2014,shenHarmonicFunctionsRotational2016,rayPractical3DFrame2016,cheminRepresentingThreeDimensionalCross2019,PBS20,golovatyVariationalMethodGenerating2021}.

However, some octahedral field defect configurations are known to be incompatible with hexahedral meshing, necessitating the development of algorithms to correct singularity graphs and associated frame fields \cite{liuSingularityconstrainedOctahedralFields2018,cormanSymmetricMovingFrames2019,liuLocallyMeshableFrame2023,coiffierMethodMovingFrames2023}. Thus far these methods only operate at the level of individual defect lines. 
As such, they cannot guarantee global meshability, i.e., the realizability of a defect configuration in a hexahedral mesh. To make theoretical progress on this problem would require a better understanding of the global topology of defect configurations in both the field and mesh settings.
\begin{figure}
\centering
\newcommand{\imgwidth}{0.3\textwidth}
\begin{tabular}{ccc}
\includegraphics[width=\imgwidth]{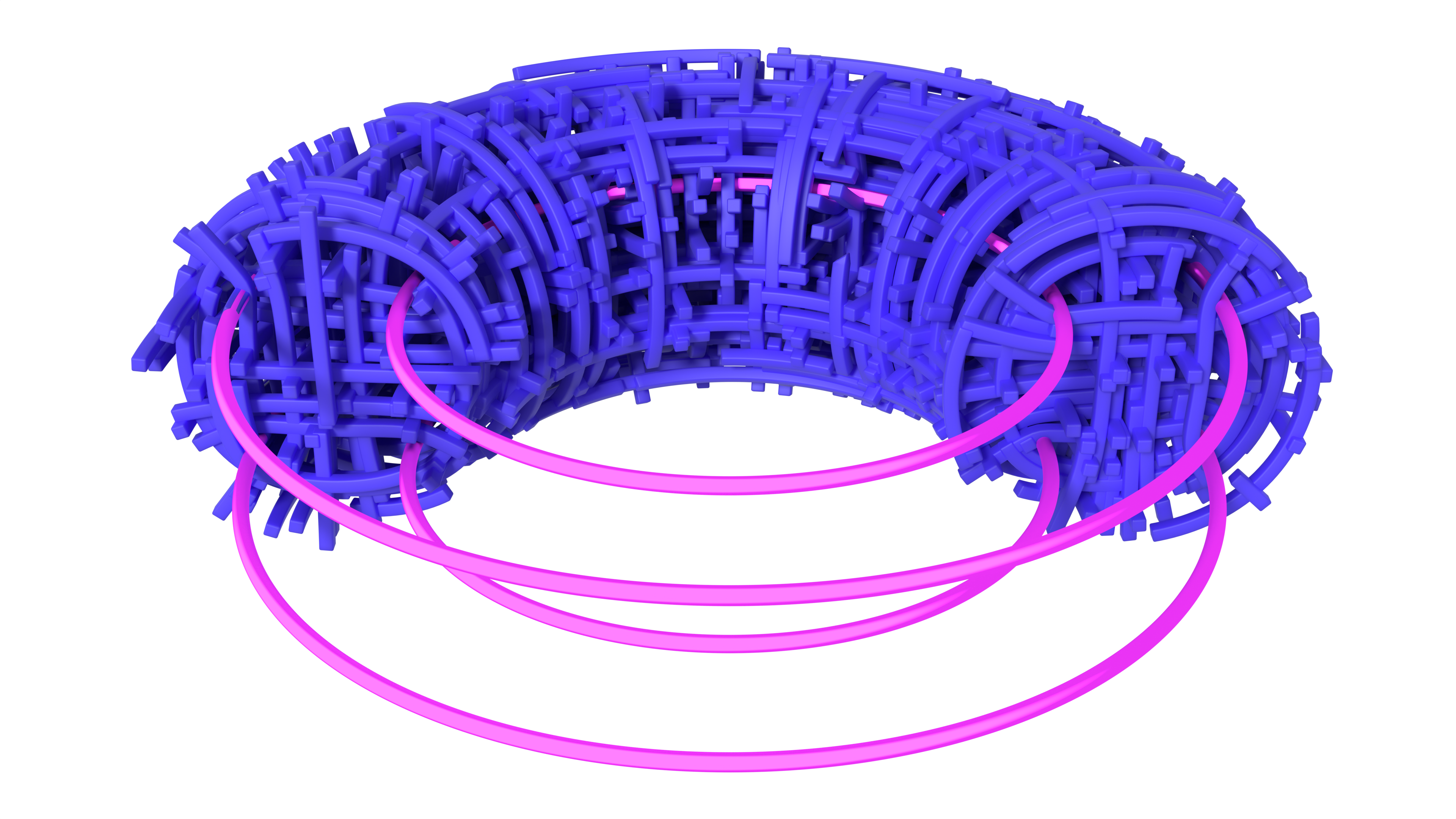} &
\includegraphics[width=\imgwidth]{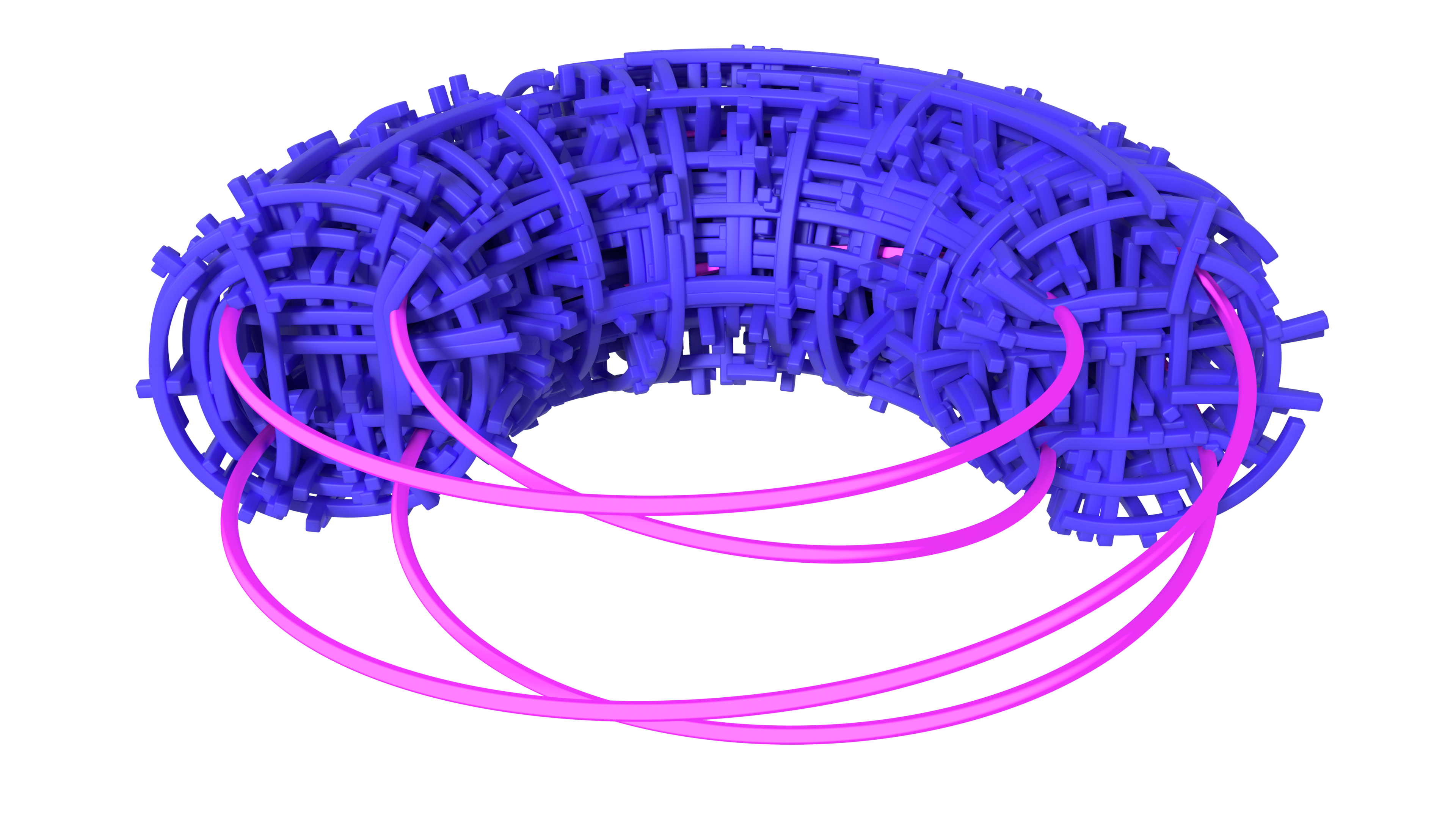} &
\includegraphics[width=\imgwidth]{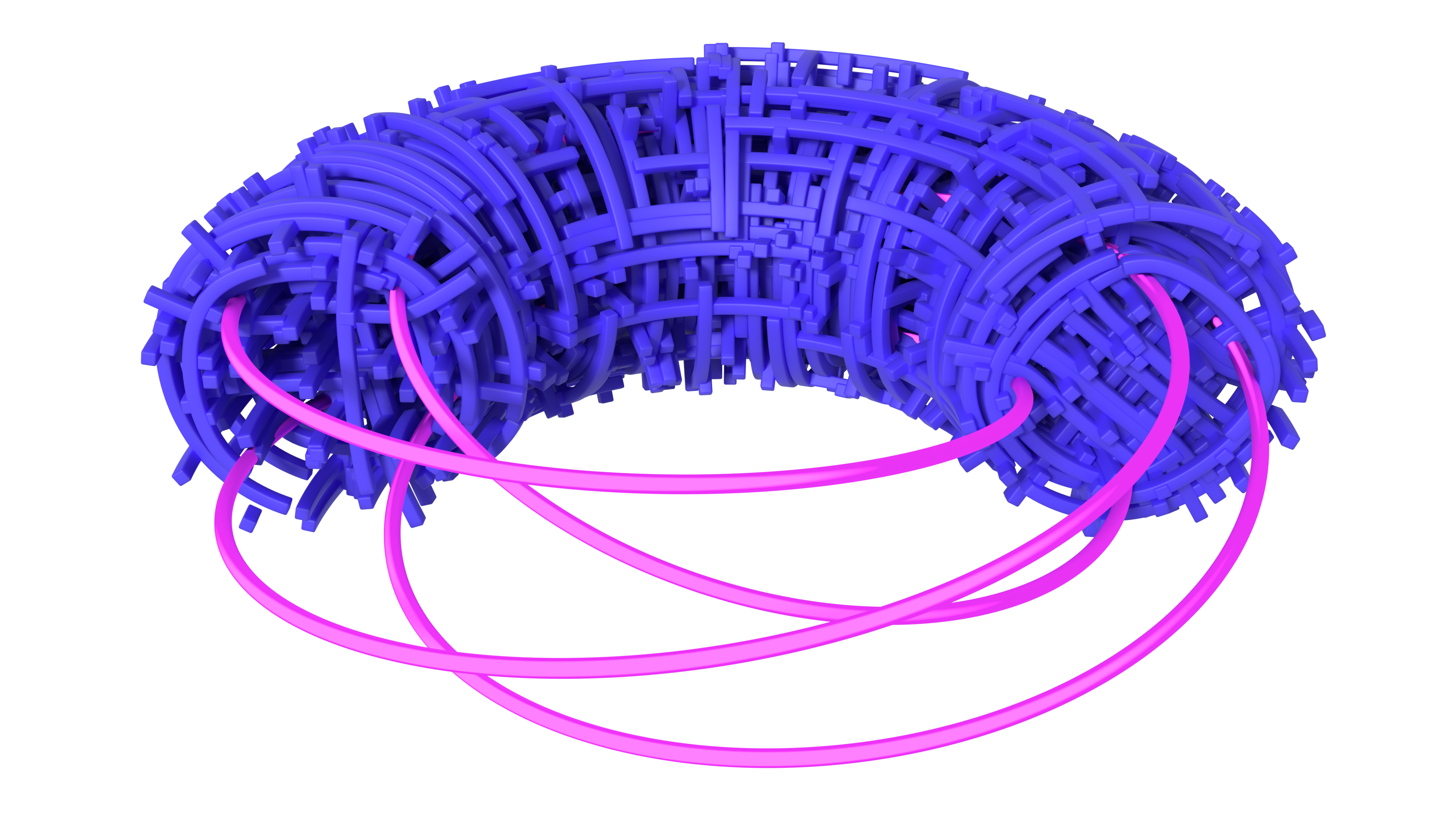} \\
\includegraphics[width=\imgwidth]{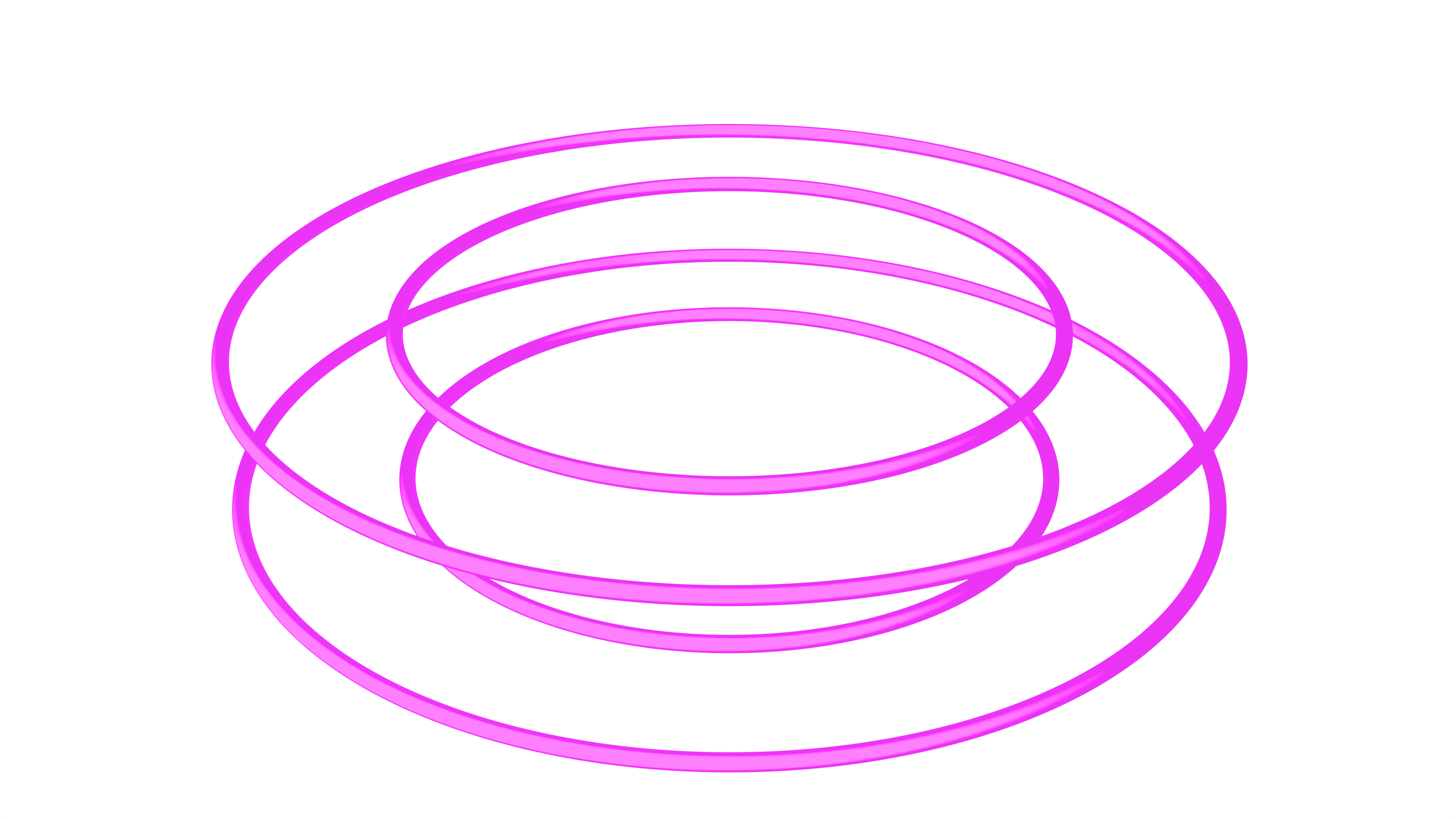} &
\includegraphics[width=\imgwidth]{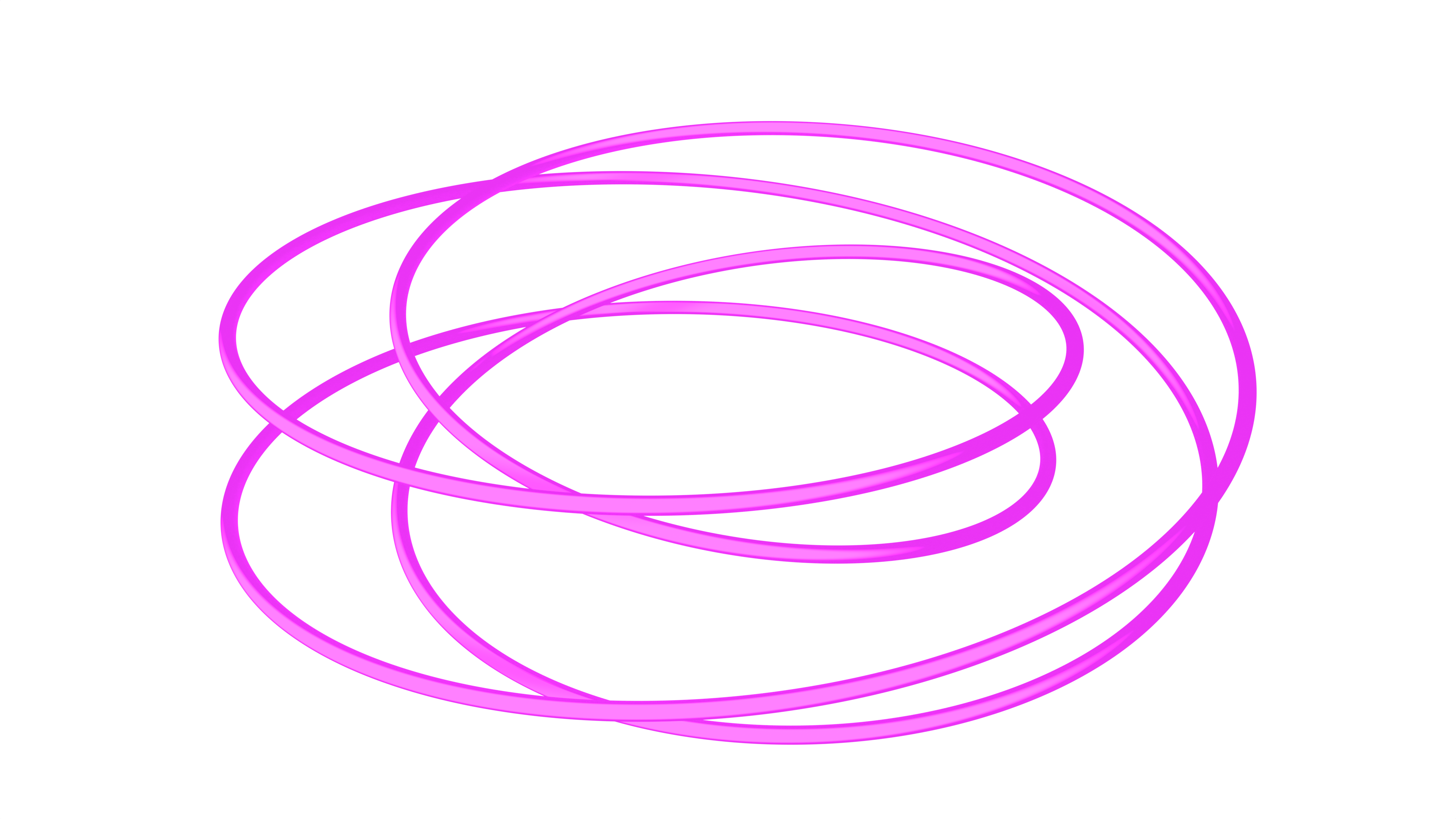} &
\includegraphics[width=\imgwidth]{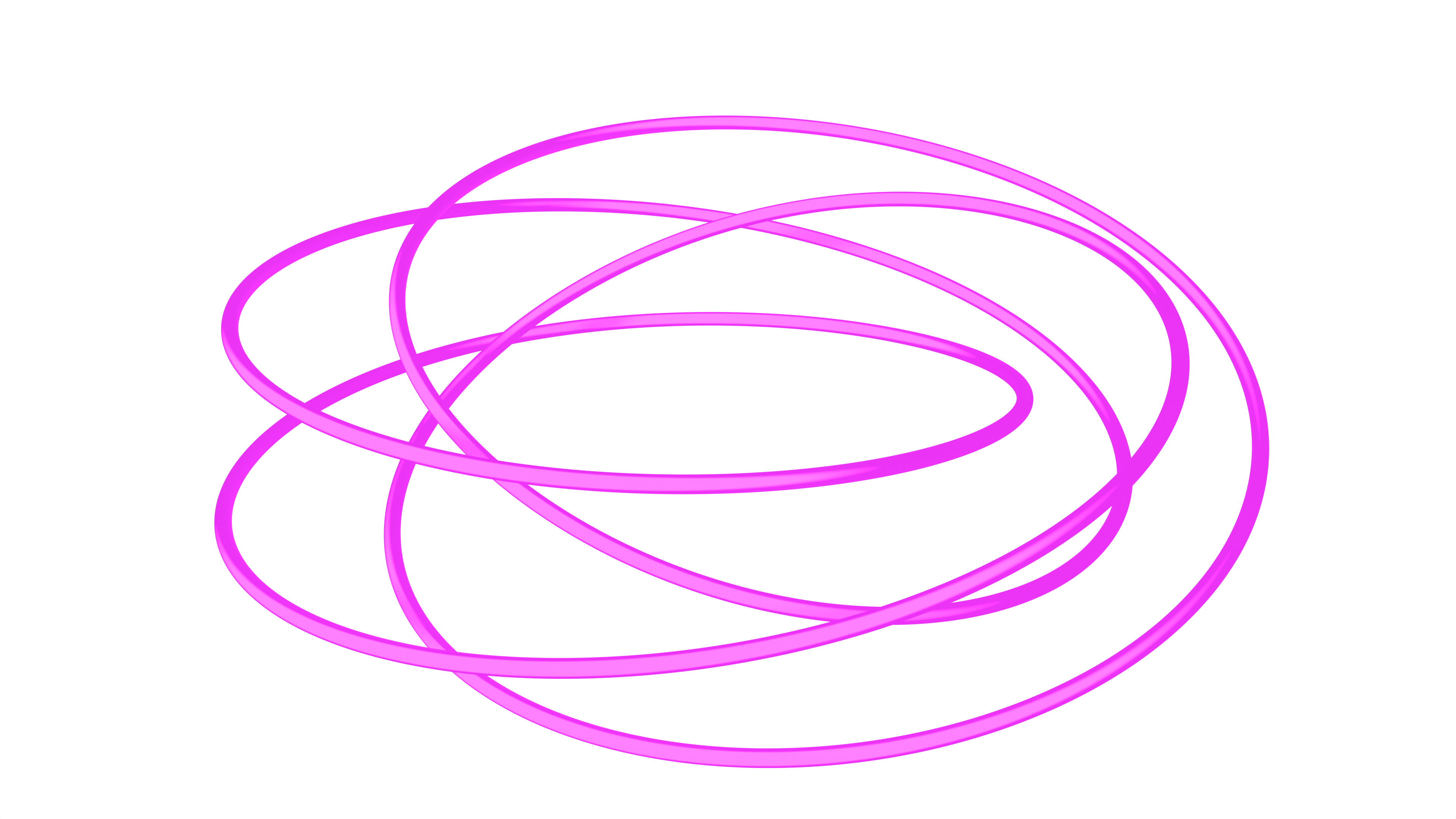} \\
$(4, 0)$ & $(4, 2)$ & $(4, 3)$
\end{tabular}
\caption{Integral curves and defects of three different octahedral frame fields on the solid torus. Their defects---all of conjugacy class $\left[(1 + i)/\sqrt{2}\right] \subset 2O$---form a $4$-component unlink (left), a $(4, 2)$-torus link (center), and a $(4, 3)$-torus knot (right).}
\label{fig:octa-examples}
\end{figure}

In applications to both liquid crystal physics and meshing, interactions between defects and the domain boundary are significant. Experimentalists have used anchoring of liquid crystals to substrates to control defects \cite{taiSurfaceAnchoringControl2020}, interacting colloidal particles \cite{senyukDesignPreparationNematic2022}, and even micro-robots \cite{yaoNematicColloidalMicroRobots2022}. 
In meshing, the problem instance is specified by the boundary geometry, and correct discretization of the boundary is vitally important for downstream applications of hexahedral meshes. Despite these practical considerations, boundary conditions complicate the classification of defect configurations, and most existing work on classification has focused on the boundary-free setting, i.e., defects in the $3$-sphere.

Topological invariants that distinguish global defect configurations have been proposed in \cite{holzTopologicalPropertiesStatic1992,AZM22,ARM24}. 
However, these invariants are not known to be sufficient to completely classify configurations. To our knowledge, the first work in that direction is by \textcite{NKTK24}, who give a homotopy classification of knotted defects either in the Euclidean space $\R^3$, the $3$-sphere $S^3$, or the closed $3$-ball.

The present work marks a first step toward a global classification in domains with boundary. We extend the tools of \cite{NKTK24} to the setting of defects in unknotted handlebodies, subject to some boundary conditions. 
More precisely, we first formulate a mathematical model for ordered media with defects on a handlebody with boundary conditions and homotopical equivalence relations between them.
We then define certain algebraic data derived from the defect configurations. 
On the other hand, the equivalence classes of ordered media can be described using \emph{colored diagrams} and their moves.
Our main result, Theorem~\ref{thm:top_classif}, 
gives natural bijective correspondences between ordered media with boundary conditions modulo equivalence; 
colored diagrams up to moves; 
and the algebraic data. 
This allows us to classify defect configurations modulo equivalence using purely algebraic data that can be derived from their diagrams.  

For concreteness, we illustrate our classification in the case of a biaxial nematic system on the solid torus $M \defeq D^2 \times S^1$ in Appendix~\ref{sec: Classification example}. 
The order parameter space of that system is $S^3 / Q$, where 
$Q = \{ \pm 1 , \pm i, \pm j, \pm k \}$ is the \emph{quaternion group}, and 
the fundamental group of the boundary $\partial M$ of $M$, which is a torus, is the free abelian group generated by two loops $\alpha$ and $\beta$ called the 
\emph{longitude} and \emph{meridian}, respectively. 
If we restrict our interest to systems without boundary defects, and further, if we impose a boundary condition that the restriction $f_0\colon \partial M \to S^3 / Q$ satisfies $(f_0)_* (\alpha) = i \in Q$, then any such system is equivalent
to exactly one of the twelve models depicted in \Cref{fig:example_classification_intro}, 
whose defect graphs are indicated in black and labeled with monodromies around edges.
\begin{figure}[htbp]
\centering\includegraphics[width=0.95\textwidth]{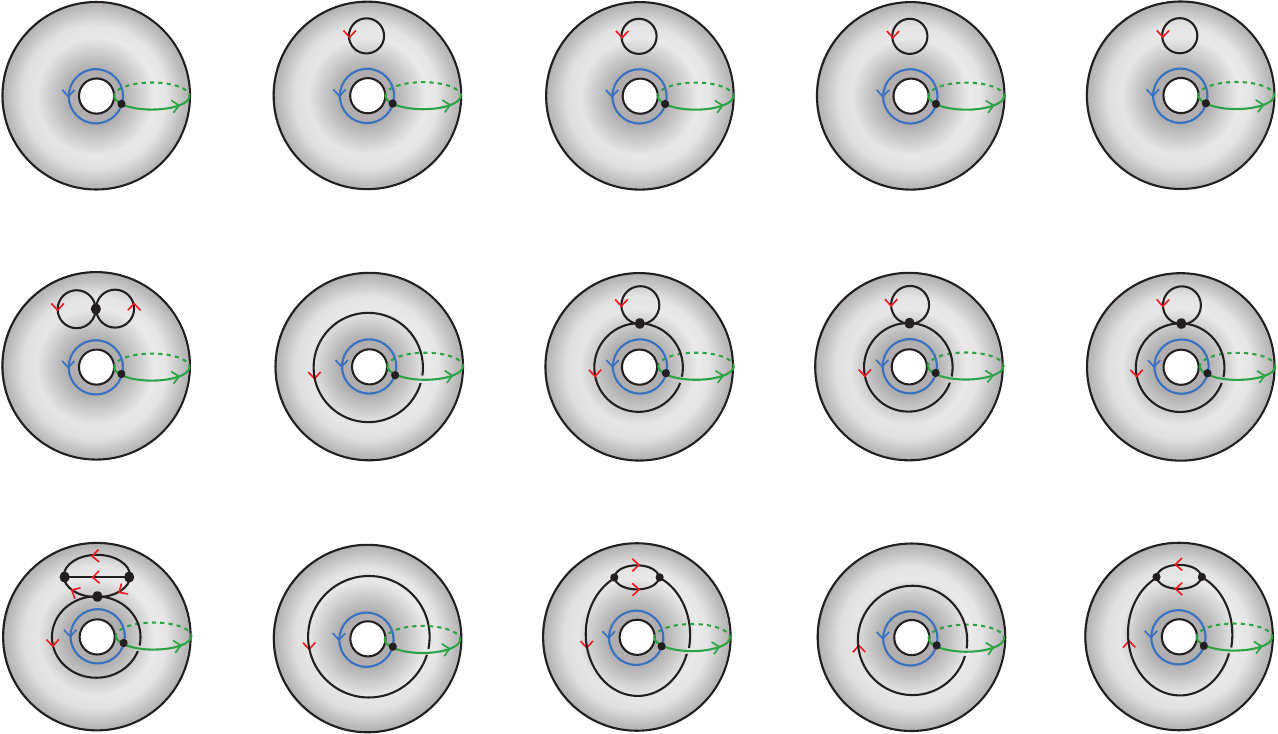}
\begin{picture}(400,0)(0,0)

\put(24,212){\footnotesize \color{blue}$\alpha$}
\put(60,212){\footnotesize \color{teal}$\beta$}

\put(95,231){\footnotesize \color{red}$-1$}
\put(110,212){\footnotesize \color{blue}$\alpha$}
\put(146,212){\footnotesize \color{teal}$\beta$}

\put(188,231){\footnotesize \color{red}$i$}
\put(196,212){\footnotesize \color{blue}$\alpha$}
\put(232,212){\footnotesize \color{teal}$\beta$}

\put(273,231){\footnotesize \color{red}$j$}
\put(282,212){\footnotesize \color{blue}$\alpha$}
\put(318,212){\footnotesize \color{teal}$\beta$}

\put(358,231){\footnotesize \color{red}$k$}
\put(368,212){\footnotesize \color{blue}$\alpha$}
\put(404,212){\footnotesize \color{teal}$\beta$}

\put(24,126){\footnotesize \color{blue}$\alpha$}
\put(60,126){\footnotesize \color{teal}$\beta$}
\put(10,145){\footnotesize \color{red}$j$}
\put(42,145){\footnotesize \color{red}$k$}

\put(110,126){\footnotesize \color{blue}$\alpha$}
\put(146,126){\footnotesize \color{teal}$\beta$}
\put(85,123){\footnotesize \color{red}$-1$}

\put(196,126){\footnotesize \color{blue}$\alpha$}
\put(232,126){\footnotesize \color{teal}$\beta$}
\put(188,145){\footnotesize \color{red}$i$}
\put(173,123){\footnotesize \color{red}$-1$}

\put(282,126){\footnotesize \color{blue}$\alpha$}
\put(318,126){\footnotesize \color{teal}$\beta$}
\put(273,145){\footnotesize \color{red}$j$}
\put(258,123){\footnotesize \color{red}$-1$}

\put(368,126){\footnotesize \color{blue}$\alpha$}
\put(404,126){\footnotesize \color{teal}$\beta$}
\put(359,145){\footnotesize \color{red}$k$}
\put(344,123){\footnotesize \color{red}$-1$}

\put(24,40){\footnotesize \color{blue}$\alpha$}
\put(60,40){\footnotesize \color{teal}$\beta$}
\put(0,37){\footnotesize \color{red}$-1$}
\put(14,53){\footnotesize \color{red}$i$}
\put(40,53){\footnotesize \color{red}$i$}
\put(29,73){\footnotesize \color{red}$j$}
\put(29,62){\footnotesize \color{red}$k$}

\put(110,40){\footnotesize \color{blue}$\alpha$}
\put(146,40){\footnotesize \color{teal}$\beta$}
\put(90,37){\footnotesize \color{red}$i$}

\put(196,40){\footnotesize \color{blue}$\alpha$}
\put(231,40){\footnotesize \color{teal}$\beta$}
\put(200,69){\footnotesize \color{red}$j$}
\put(200,51){\footnotesize \color{red}$k$}
\put(172,37){\footnotesize \color{red}$i$}

\put(283,40){\footnotesize \color{blue}$\alpha$}
\put(318,40){\footnotesize \color{teal}$\beta$}
\put(264,37){\footnotesize \color{red}$i$}

\put(369,40){\footnotesize \color{blue}$\alpha$}
\put(404,40){\footnotesize \color{teal}$\beta$}
\put(350,37){\footnotesize \color{red}$i$}
\put(372,68){\footnotesize \color{red}$k$}
\put(372,52){\footnotesize \color{red}$j$}
\end{picture}
\caption{A complete list of defect configurations, up to homotopy, of the biaxial nematic system on the solid torus without boundary defects and subject to the condition that the map $f_0\colon \partial M \to S^3 / Q$ defined for the system satisfies $(f_0)_* (\alpha) = i \in Q$.}
\label{fig:example_classification_intro}
\end{figure}
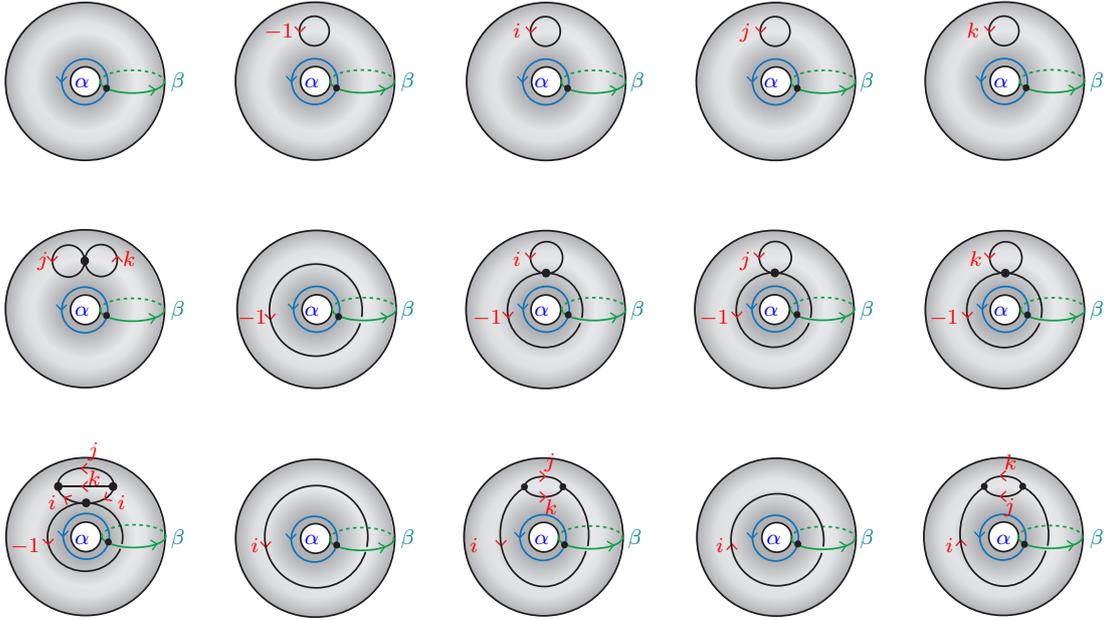

Our model can address both ordered media within $3$-dimensional handlebodies and ordered media on the exteriors of unknotted handlebodies in the $3$-dimensional Euclidean space $\R^3$, 
where the order parameter remains constant at large distances from the origin (see Remark~\ref{rem:constant far-field}).
We leave to future work the extension of these ideas to 
%fully general handlebodies 
more general subspaces of $\R^3$ and less-restrictive boundary conditions (see Remark~\ref{rem:boundary-conds}).

\paragraph{Outline.}
The structure of the paper is as follows.
Section~\ref{sec:preliminaries} reviews prior work on the classification of knotted defects in unbounded domains, following a brief introduction of fundamental tools from algebraic topology that will be used in this paper.
Section~\ref{sec:Mathematical model} develops a mathematical model of ordered media up to homotopy in bounded domains, describes their diagrammatic representation, and introduces a certain algebraic invariant.
Section~\ref{sec:main_result} presents the main result, offering a complete classification of defect configurations up to equivalence, along with a simple illustrative example.
Section~\ref{sec:proof_of_theorem} gives the proof of the main theorem.
Appendix~\ref{sec: Classification example} offers concrete examples illustrating the classification process using the main theorem.
Appendix~\ref{sec: The subgroups of the binary octahedral group} presents a classification of subgroups of the binary octahedral group $2O$ and their conjugacy classes, which can be used for the classification of octahedral frame fields.

\subsection*{Acknowledgments}
The authors would like to thank Ivan I. Smalyukh, Tam\'{a}s K\'{a}lm\'{a}n, and Shun Wakatsuki for valuable discussion.
Author Y.~Nozaki was supported by JSPS KAKENHI Grant Numbers JP20K14317, JP23K12974 and JP24H00686.
Author D.~Palmer was supported by NSF MSPRF Grant Number 2303403.
Author Y.~Koda was supported by JSPS KAKENHI Grant Numbers JP20K03588, JP21H00978, JP23H05437 and JP24K06744.

%%%
\section{Preliminaries}
\label{sec:preliminaries}
In this section, we establish the basic notation and mathematical background that will be used in the proof of our main result.

%%%
\subsection{Groups and their presentations}
\label{subsec:group_presentation}
We recall some basic definitions and facts about groups. Let $G$ be a group with identity element $1$, and let $S$ be a subset of $G$. 
We denote by $\langle S \rangle$ the subgroup of $G$ generated by $S$; 
that is, $\langle S \rangle$ is the subgroup of all elements of 
$G$ that can be expressed as finite products of elements in $S$. 
Let $\angg{S}_{G}$ denote the \emph{normal closure} of $S$ in $G$, that is, the smallest normal subgroup of $G$ containing $S$.
We say that $G$ is \emph{finitely generated} if there exists a finite subset $S \subset G$ with $\langle S \rangle = G$.
We denote the set of finitely generated subgroups of $G$ by $\SS_G^\fg$.

We say that $G$ has a \emph{presentation} 
$\langle\, S \mid R \,\rangle$ if  
$S$ is a set, $R$ is a subset of the free group $F_S$ 
on $S$, and 
$G$ is isomorphic to the quotient of $F_S$ by the normal closure $\angg{R}_{F_S}$ of $R$. 
Here $S$ and $R$ are called sets of \emph{generators} and \emph{relators}, respectively. 
Thus, if we regard the set of generators $S$ 
as a subset of $G$ by the above isomorphism, we have 
$\langle S \rangle = G$, and $R$ is a set of products of generators, each equal to $1$, that are sufficient to express all relations among the generators. 
Two sugbroups $H$ and $H'$ of a group $G$ are said to be \emph{conjugate} 
if there exists an element $x$ of $G$ with $x H x^{-1} = H'$.

Let $C_G(S)$ denote the \emph{centralizer} of a subset $S$ of $G$; that is, $C_G(S)\coloneqq \{g\in G\mid 
\forall s\in S, \; gs=sg \}$.
Given two subgroups $H$ and $K$ of $G$, the set of \emph{double cosets} $H \backslash G / K$ consists
of equivalence classes $[g]$ of elements of $G$, where $[g] = [g']$ if and only if $hgk = g'$ for some $h \in H$ and $k \in K$.

%%%
\subsection{Homotopy groups}
\label{subsec:homotopy}
Recall that the fundamental group $\pi_1(X, x_0)$ of a topological space $X$ with a 
base point $x_0$ is defined to be the group of homotopy classes of loops 
$[0,1] \to X$ based at $x_0$, where the product of 
two elements is defined by their concatenation. 

For pairs of topological spaces $(X,A)$ and $(Y,B)$, let $\Map((X,A),(Y,B))$ denote the set of continuous maps $f\colon X\to Y$ satisfying $f(A)\subset B$.
In particular, fixing base points $x_0 \in X$ and $y_0 \in Y$, 
the set $\Map((X,\{x_0\}),(Y,\{y_0\}))$ of basepoint-preserving continuous maps is simply denoted by $\Map_0(X,Y)$.
Let $\Hom(G,H)$ denote the set of homomorphisms from a group $G$ to $H$.
There is a natural map $\Map_0(X,Y)\to \Hom(\pi_1(X),\pi_1(Y))$, which takes a map $f\colon X \to Y$ to the induced homomorphism 
$f_\ast \colon \pi_1(X, x_0) \to \pi_1(Y, y_0)$ between the fundamental groups. 
This map descends to
\[
\theta\colon \Map_0(X,Y)/{\simeq}\to \Hom(\pi_1(X),\pi_1(Y)),
\]
where $\simeq$ is the equivalence relation of basepoint-preserving homotopy.

\begin{example}
For convenience, we describe a few specific groups relevant to the topology of defects in ordered media. Following \cite{Mer79}, these groups arise as fundamental groups of order parameter spaces given as quotients $\SO(3)/H$ of the $3$-dimensional rotation group $\SO(3)$ by finite subgroups (point groups) $H$. 
In such cases, $\pi_1(\SO(3)/H) = 2H \subset \SU(2)$, a central extension of $H$ by $\Z_2$.
\begin{itemize}
\item The \emph{quaternion group} 
$Q=\{\pm 1, \pm i, \pm j, \pm k\}$ is the finite non-abelian group generated by four elements $-1, i, j, k$ such that 
$-1$ commutes with other elements and 
they satisfy $(-1)^2 = 1$ and $i^2 = j^2 = k^2 = ijk = -1$.
Note that $Q = 2K$, where 
$K$ is the \emph{Klein four-group} $\Z_2 \times \Z_2$. 
The group $Q$ classifies line defects in \emph{biaxial nematics}.
\item The \emph{binary tetrahedral group} $\BTet$ is obtained by adding 
sixteen elements
\[ c^{\pm 1}, c^{\pm 2}, \alpha^{\pm 1}, \alpha^{\pm 2}, \beta^{\pm 1}, \beta^{\pm 2}, \gamma^{\pm 1} , \gamma^{\pm 2} \]
to the quaternion group $Q$, 
where $c \defeq \frac{1}{2}(1 + i + j + k)$, $\alpha \defeq \frac{1}{2}(1 + i - j - k)$, 
$\beta \defeq \frac{1}{2}(1 - i + j - k)$ and $\gamma \defeq \frac{1}{2}(1 - i - j + k)$.
\item The \emph{binary octahedral group} $\BOct$ is obtained by adding 
$24$ elements to $\BTet$:
\[
\frac{1}{\sqrt{2}}(\pm 1 \pm i),\  \frac{1}{\sqrt{2}}(\pm 1 \pm j),\ 
\frac{1}{\sqrt{2}}(\pm 1 \pm k),\ \frac{1}{\sqrt{2}}(\pm i \pm j),\ 
\frac{1}{\sqrt{2}}(\pm j \pm k),\ \frac{1}{\sqrt{2}}(\pm k \pm i).
\]
$\BOct$ classifies line defects in \emph{octahedral frame fields}.
\end{itemize}
In this way, we have a stratification $Q=2K \subset \BTet \subset \BOct$. 
One can check that every element of $Q$ can be expressed as a product of powers of $i$ and $j$; thus, we can write $Q = \langle i,j \rangle$. 
Similarly, we have $\BTet = \langle i, c \rangle$ and $\BOct = \left\langle c, \frac{1}{\sqrt{2}} (i+j) \right\rangle$. 
Note that $\langle i \rangle$ and $\langle j \rangle$ are not conjugate 
as subgroups of $Q$, whereas they are conjugate as subgroups of $\BTet$ (and so of $\BOct$), for we have $c \langle i \rangle c^{-1} = \langle j \rangle$. 
See Appendix~\ref{sec: The subgroups of the binary octahedral group} for the complete list of subgroups of $\BOct$ and their conjugacy classes.
\end{example}

For an integer $n\geq 2$, the $n$th homotopy group $\pi_n(X,x_0)$ is obtained from $\Map_0(S^n,X)$ by considering basepoint-preserving homotopies and a concatenation operation similar to that of $\pi_1$.
In contrast to $\pi_1$, these groups are always abelian.
See, e.g., Hatcher \cite{Hat02}, for more details. 

%%%
\subsection{Wirtinger presentation}
\label{subsec:Wirtinger_presentation}

\begin{definition}
By a \emph{graph} we mean a finite $1$-dimensional CW complex, where 
$0$-cells and $1$-cells are called its \emph{vertices} and \emph{edges}, respectively. 
A subspace $\Gamma$ of $S^3$ (or $\R^3$) is called a  \emph{spatial graph} if 
$\Gamma$ is ambient isotopic to the union of a finite number of straight line segments, where 
$\Gamma$ is endowed with the structure of a graph. 
\end{definition}
Note that knots and links are examples of spatial graphs. 
Note also that the same subset of $S^3$ (or $\R^3$) admits distinct graph structures, distinguished by vertices of valence two. 
Two spatial graphs are said to be \emph{equivalent} if they are related by an ambient isotopy preserving the graph structure.

Fix a projection $p\colon \R^3 \to \R^2$. 
A spatial graph $\Gamma \subset \R^3$ is said to be in a \emph{regular position} with respect to $p$ if 
\begin{enumerate}[label=(\arabic*)]
\item 
for each point $x$ in $\R^2$, the set $p^{-1} (x) \cap \Gamma$ consists of at most two points; 
\item
there are only finitely many points $x$ of $\R^2$ with 
$| p^{-1} (x) \cap \Gamma | = 2$; 
such points $x$ are called \emph{crossings}; 
\item
no vertex of $\Gamma$ projects to a crossing; and 
\item
each crossing $x \in \R^2$ has a neighborhood $U_x \subset \R^2$ such that the set 
$p^{-1} (U_x) \cap \Gamma$ consists of two straight line intervals.
\end{enumerate}
When we consider spatial graphs in $S^3$ instead of in $\R^3$, we use the projection $S^3 \to S^2$ 
induced from $p\colon \R^3 \to \R^2$ by the one-point compactification. 
Note that every spatial graph can be moved by an isotopy to be in a regular position with respect to $p$. 
Let $\Gamma$ be a spatial graph in a regular position with respect to $p$. 
A \emph{diagram} $D_\Gamma$ of $\Gamma$ is defined to be the image $p (\Gamma )$ such that 
at each crossing, the incident arcs are labeled with their relative height (over/under). 
Thus, a diagram of a spatial graph is a purely combinatorial object consisting of 
a planar graph on $\R^2$ with over/under information at some of its vertices. 
Diagrams of equivalent spatial graphs can be related by a finite sequence of local moves called \emph{Reidemeister moves}, which are illustrated in \Cref{fig:local_move_2} (in the case of colored diagrams).

Let $\Gamma \subset S^3$ be a spatial graph, and 
take $x_0 \in S^3 \setminus \Gamma$ as a base point. 
Then, we can get a finite presentation of the fundamental group 
$\pi_1 (S^3 \setminus \Gamma)$ 
from a diagram $D_\Gamma$ as follows. 
For simplicity, we assume that the diagram is connected. 
Fix an orientation of each edge of $\Gamma$. 
The diagram of $\Gamma$ consists of finitely many oriented arcs $\alpha_1, \ldots, \alpha_n$, where the endpoints of each $\alpha_i$ are 
vertices or crossings. 
For $i=1, \ldots, n$, assign a letter $x_i$ to each oriented arc $\alpha_i$. 
At each crossing or vertex $c_j$ ($j=1, \ldots, m$), 
define a relator $r_j$ following the 
rule shown in Figure \ref{fig:wirtinger}, 
where at each vertex $v$, 
$\varepsilon_{i_r}=1$ (resp.\  $\varepsilon_{i_r}=-1$) if the 
$r$th half-edge incident to $v$ is a tail (resp.\ head). 
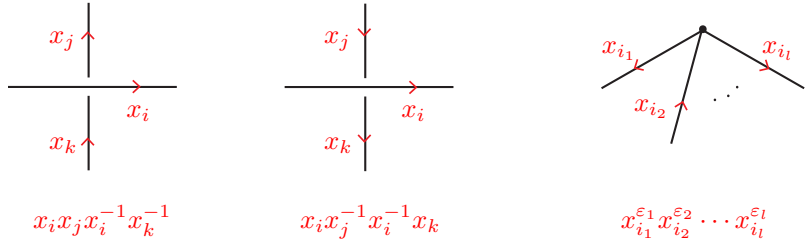
\begin{figure}[htbp]
\centering\includegraphics[width=0.70\textwidth]{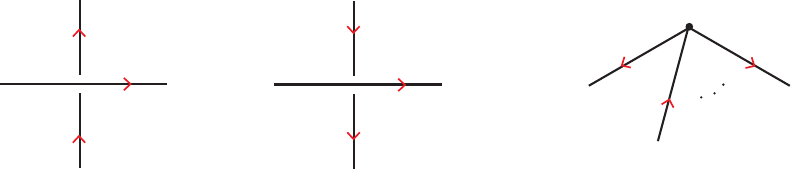}
\begin{picture}(400,0)(0,0)
\put(95,42){\color{red} $x_i$}
\put(66,70){\color{red} $x_j$}
\put(66,30){\color{red}
$x_k$}
\put(60,0){\color{red} $x_i x_j x_i^{-1} x_k^{-1}$}

\put(198,42){\color{red} $x_i$}
\put(169,70){\color{red} $x_j$}
\put(169,30){\color{red} $x_k$}
\put(160,0){\color{red} $x_i x_j^{-1} x_i^{-1} x_k$}

\put(273,66){\color{red} $x_{i_1}$}
\put(285,44){\color{red} $x_{i_2}$}
\put(333,66){\color{red} $x_{i_l}$}
\put(280,0){\color{red} $x_{i_1}^{\varepsilon_1}  x_{i_2}^{\varepsilon_2} \cdots x_{i_l}^{\varepsilon_l}$}
\end{picture}
\caption{Relators in a Wirtinger presentation.}
\label{fig:wirtinger}
\end{figure}
The group $\pi_1 (S^3 \setminus \Gamma, x_0)$ then has a presentation $\langle x_1, \ldots, x_n \mid r_1, \ldots, r_m \rangle$, which is called a \emph{Wirtinger presentation}. 
In this paper, concatenations of loops in $\pi_1$ are read from left to right, the basepoints of all diagrams are above the page, and orientations of meridians follow the right-hand rule.
See Kawauchi~\cite{Kaw96} for more details.

%%%
\subsection{Review of previous work}
\label{subsec:previous_work}
We recall the mathematical model of global defect configurations introduced in \cite{NKTK24}.
Let $G$ be a group, and let $X_G$ be a CW complex satisfying $\pi_1(X_G)\cong G$ and $\pi_2(X_G)=\{0\}$.
Note that such a space $X_G$ always exists for any $G$. 
Following standard terminology in physics, we refer to $X_G$ as the \emph{order parameter space}.

Fix basepoints $p_0 \in \R^3$ and $x_0 \in X_G$.
\begin{definition}
A \emph{defect configuration} $(\Gamma, f)$ comprises a spatial graph $\Gamma$, the \emph{defect set}, together with an \emph{order parameter field} $f \in \Map_0(\R^3 \setminus \Gamma, X_G)$ defined on the complement of $\Gamma$.	
\end{definition}

We classify defect configurations up to an equivalence relation generalizing homotopy.
For defect sets $\Gamma, \Gamma'$ and order parameter fields $f\in \Map_0(\R^3\setminus \Gamma, X_G)$ and $f'\in \Map_0(\R^3\setminus \Gamma', X_G)$, we say $(\Gamma, f) \sim (\Gamma', f')$, if, after adding finitely many vertices $v_i$, $v'_j$ and edges $e_k$, $e'_l$ to each graph, there exists a basepoint-preserving ambient 
isotopy $\{h_t\}_{t\in [0,1]}$ of $\R^3$ sending $\widehat{\Gamma}=\Gamma\cup\{v_i\}\cup\{e_k\}$ to 
$\widehat{\Gamma}'=\Gamma'\cup\{v'_j\}\cup\{e'_l\}$ in such a way that $f|_{\R^3\setminus \widehat{\Gamma}}$ is homotopic to $f'|_{\R^3\setminus \widehat{\Gamma}'}\circ h_1$.

\begin{theorem}[\cite{NKTK24}]
\label{thm:NKTK24}
The map 
\[
\Phi\colon \left(\coprod_\Gamma \Map_0(\R^3\setminus \Gamma, X_G)\right)/{\sim} \to \SS^\fg_G
\]
defined by $\Phi([f])=\Im(f_\ast\colon \pi_1(\R^3\setminus\Gamma)\to G)$,
%where $\SS_G^\fg$ denotes the set of finitely generated subgroups of $G$, 
is bijective.
\end{theorem}

%%%%%%%%%
\section{Mathematical model}
\label{sec:Mathematical model}

Let $k$ be a non-negative integer. 
A \emph{handlebody} of genus $k$ is a compact orientable 3-manifold with non-empty boundary obtained by attaching 
$k$ 1-handles to a $0$-handle. 
Note that a handlebody of genus $0$ is just a closed $3$-ball, 
while a handlebody of genus $1$ is a solid torus $D^2 \times S^1$.  
See \Cref{fig:handlebodies}. 
In this section, we extend the mathematical model reviewed in Section~\ref{subsec:previous_work} to cover defects in handlebodies.

\begin{figure}[h]
    \centering
    \includegraphics[width=0.8\textwidth]{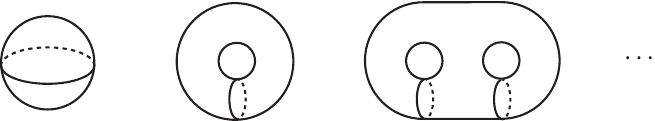}
\begin{picture}(400,0)(0,0)
\put(38,0){genus 0}
\put(137,0){genus 1}
\put(255,0){genus 2}
\end{picture}
    \caption{Handlebodies.}
    \label{fig:handlebodies}
\end{figure}
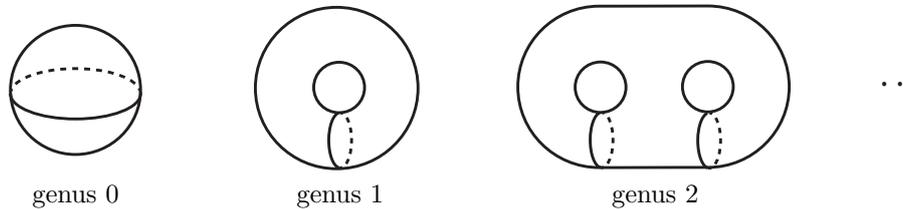

Let $W$ be an unknotted handlebody of genus $k$ in $S^3$; that is, a handlebody whose exterior $M=E(W)=S^3\setminus\Int W$ is also a handlebody. 
For example, if we regard the handlebodies shown in Figure~\ref{fig:handlebodies} 
as embedded in $S^3$, then they are unknotted. 
An unknotted handlebody of genus $1$ can be obtained as a regular neighborhood of the unknot, and \emph{knotted} handlebodies of genus $1$ can be obtained as regular neighborhoods of non-trivial knots in $\R^3$.
For convenience, we work on the ``unbounded'' handlebody $M$ and treat its complement as the exterior.

%%%
\subsection{Formulation via continuous maps}
Fix basepoints $\ast_\partial\in (\partial M)\setminus P$ and $\ast \in M$ for the boundary and interior, respectively.
Also fix a set of \emph{boundary defects} $P=\{p_1,\dots,p_n\}\subset \partial M$ and \emph{boundary values} $f_0\in \Map_0((\partial M)\setminus P, X_G)$ such that $f_0$ cannot extend to any $p_i$.

We consider spatial graphs $\Gamma$ in $M \setminus \{\ast\}$ terminating at boundary defects---i.e., satisfying $\partial M\cap \Gamma = P\cap \partial \Gamma$.
For such a graph, define the associated set of order parameter fields satisfying boundary conditions
\[
\Map_0(M \setminus \Gamma, X_G; f_0) = \{f\in \Map_0(M \setminus \Gamma, X_G) \mid \text{$f|_{\partial M \setminus P}$ is homotopic to $f_0$}\}.
\]
\begin{definition}
\label{def: equivalence relation in our mathematical model}
We take the total set of defect configurations to be the disjoint union $\coprod_\Gamma \Map_0(M \setminus \Gamma, X_G; f_0)$, where $\Gamma$ runs over spatial graphs as above.
We consider two maps $f\in \Map_0(M \setminus \Gamma, X_G; f_0)$ and $f'\in \Map_0(M \setminus \Gamma', X_G; f_0)$ to be \emph{equivalent}, $f\sim f'$, if there exist finitely many vertices $\{v_i\}$, $\{v'_k\}$ and edges $\{e_j\}$, $\{e'_l\}$, and a basepoint-preserving ambient isotopy $\{h_t\}_{t\in [0,1]}$ of $M$, such that:
\begin{itemize}
	\item $h_t$ restricts to the identity map on $\partial M$;
	\item $h_t$ sends $\widehat{\Gamma}\coloneqq\Gamma\cup\{v_i\}\cup\{e_j\}$ to $\widehat{\Gamma}'\coloneqq\Gamma'\cup\{v'_k\}\cup\{e'_l\}$; and
	\item $f|_{M \setminus \widehat{\Gamma}}$ is homotopic to $f'|_{M \setminus \widehat{\Gamma}'}\circ h_1$ relative to basepoints.
\end{itemize}
\end{definition}

\begin{remark}
\label{rem:boundary-conds}
The boundary conditions we consider here differ subtly from those that might be more familiar from applications. 
While in many applications boundary defects are free to move along the boundary, our model strictly fixes the positions of the boundary defects $P$.
	In contrast to this strong constraint on boundary defect positions, our treatment of boundary values is relatively loose:
	we consider them fixed only up to homotopy through the full order parameter space $X_G$.
	Some applications feature ``anchoring'' boundary conditions, which constrain boundary values to a subbundle of
	the full order parameter bundle. For example, a common condition on nematics or octahedral frames constrains them to align to the boundary normals.
	Such ``subbundle'' boundary conditions are inexpressible in our model. 
 Nevertheless, we believe that our model will serve as a bridge to understanding these more familiar boundary conditions as our model is geometrically meaningful and allows complete classification of defect configurations as in the boundary-free case.
\end{remark}

Our classification relates defect configurations to subgroups of $G$ augmented with additional data. 
Given subsets $S, S'$ of $G$ with $S \supset S'$, we write $\SS_{G,S,S'}$ for the set of triples $(H,N,[g])$ such that
\begin{enumerate}[label=(\arabic*)]
    \item $H$ is a subgroup of $G$, 
    $N$ is a normal subgroup of $H$, and 
    $g$ is an element of $G$ such that $gSg^{-1}\subset H$ and $gS'g^{-1}\subset N$;
    \item there exists a finite subset $T$ of $N$ satisfying $\ang{gSg^{-1}\cup T}=H$ and $\angg{T}_H=N$; and
    \item $[g]$ is a double coset represented by $g$ in $N\backslash G/C_G(S)$.
\end{enumerate}
Note that the conditions $gSg^{-1}\subset H$ and $gS'g^{-1}\subset N$ do not depend on the choice of a representative of $[g]$. 
When $G$ is a finite group, we can drop (2) from the above definition by letting $T=N$.

Intuitively, given a set of boundary monodromies $S$, the subgroup $H$ encodes all monodromies both around defects and along non-contractible loops on $\partial M$, the normal subgroup $N$ of $H$ captures only 
monodromies around interior defects,
and $[g]$ records the choice of path connecting the boundary and interior basepoints. 
The necessity of all three data will be demonstrated explicitly in \Cref{ex:meaning_of_g}.
Later, we will mainly consider the case where $S$ is a finitely generated subgroup; then (2) implies that $H$ is also finitely generated.
In particular, $\SS_{G,\{1\},\{1\}}$ corresponds to $\SS^\fg_G$ in Section~\ref{subsec:previous_work} since $H=N$ and $[g]=[1]$ for any $(H,N,[g])\in \SS_{G,\{1\},\{1\}}$.

\begin{example}
\label{example: triple}
Let $G$ be the quaternion group $Q$, 
$S=\langle i \rangle = \{ \pm 1, \pm i \}$, 
and $S' = \{ 1 \}$. 
Let $(H, N, [g])$ be an element of 
$\SS_{Q,\langle i \rangle,\{1\}}$. 
Since 
$C_G (S) = C_Q ( \langle i \rangle ) = \langle i \rangle$, 
when $N = \langle j \rangle, \langle k \rangle$, or $Q$, the set $N \backslash Q / C_Q ( \langle i \rangle )$ consists of a single coset $Q$. 
On the other hand, when $N = \{ 1 \}, \langle -1 \rangle$, or $\langle i \rangle$, the set $N \backslash Q / C_Q ( \langle i \rangle )$ consists of exactly two double cosets 
$[1] = \{ \pm 1, \pm i  \}$ and 
$[j] = \{ \pm j, \pm k  \}$.

Since $H$ should contain a conjugate of 
$S = \langle i \rangle$, 
$H$ is $\langle i \rangle$ or $Q$. 
When $H = \langle i \rangle$, 
$(H, N, [g])$ is one of the six elements 
$(\langle i \rangle, \{ 1 \} , [1]) $, 
$(\langle i \rangle, \{ 1 \} , [j]) $, 
$(\langle i \rangle, \langle -1 \rangle , [1]) $
$(\langle i \rangle, \langle -1 \rangle , [j])$, 
$(\langle i \rangle, \langle i \rangle , [1]) $, and  
$(\langle i \rangle, \langle i \rangle , [j])$; when $H = Q$,  
$(H, N, [g])$ is one of the three elements 
$(Q, \langle j \rangle , [1]) $, 
$(Q, \langle k \rangle , [1])$ and 
$(Q, Q , [1]) $. 
In this way, we can obtain the complete list of elements of $\SS_{Q,\langle i \rangle,\{1\}}$. 
Note that $(Q, \langle i \rangle, [1] )$, for example, is not an element of $\SS_{Q,\langle i \rangle,\{1\}}$,  for we cannot find a subset $T$ of $\langle i \rangle$ satisfying 
$\langle \langle i \rangle \cup T \rangle = Q$. 
The following lists of elements of $\SS_{Q,\langle i \rangle, \langle -1 \rangle }$ and 
$\SS_{Q,\langle i \rangle,\langle i \rangle}$
can also be obtained in the same way: 
\begin{align*}
\SS_{Q,\langle i \rangle, \langle -1 \rangle}  &= \left\{
\begin{array}{l}
(\langle i \rangle , \langle -1 \rangle, [1]), ~
(\langle i \rangle , \langle -1 \rangle, [j]), ~
( \langle i \rangle , \langle i \rangle , [1] ),~ 
( \langle i \rangle , \langle i \rangle , [j] ), \\
( Q , \langle j \rangle , [1] ), ~
( Q , \langle k \rangle , [1] ), ~
( Q , Q , [1] )
\end{array}
\right\}, \\
\SS_{Q,\langle i \rangle, \langle i \rangle}  
&= \left\{ 
( \langle i \rangle , \langle i \rangle , [1] ), ~
( \langle i \rangle , \langle i \rangle , [j] ), ~
( Q , Q , [1] )
\right\} .
\end{align*}
\end{example}

Let $f\in \Map_0( M \setminus \Gamma, X_G; f_0)$.
For each edge $e$ of $\Gamma$, let $\mu_e$ be a meridian of $e$, that is, a boundary loop of a sufficiently small disk intersecting $e$ transversely at a point. 
Connecting $\mu_e$ to $\ast$ defines $f_\ast(\mu_e)$ up to conjugation in $\Im f_\ast$; consequently,
the normal closure $N_f$ of $\{f_\ast(\mu_e)\mid \text{$e$ is an edge of $\Gamma$}\}$ in $\Im f_\ast$
is well-defined.
Let $N_{\beta\gamma}$ be the normal closure of $\{(f_0)_\ast(\gamma)\mid \gamma\in \{\beta_1,\dots,\beta_{k}, \gamma_1,\dots,\gamma_n\}\}$ in $\Im (f_0)_\ast$, where $\beta_i$, $\gamma_j$ form a subset of basis loops for $\pi_1(\partial M \setminus P)$ as depicted in \Cref{fig:loops_beta_gamma}. 
Intuitively, $N_f$ encodes the monodromies around defects, and $N_{\beta\gamma}$ collects the monodromy data imposed by the boundary conditions on loops that would be contractible in $M$ if not for the defects.
%, where $f_\ast\colon \pi_1( M \setminus \Gamma,\ast)\to G$ is the induced homomorphism.
Then, we define a map
\[
\Phi\colon \Biggl(\coprod_\Gamma \Map_0( M \setminus \Gamma, X_G; f_0)\Biggr)/{\sim} \to \SS_{G,\Im(f_0)_\ast,N_{\beta\gamma}}
\]
by $\Phi([f])=(\Im f_\ast, N_f, [g])$.
Here $[g]\in N_f\backslash G/C_G(\Im(f_0)_\ast)$ is a unique element satisfying $g((f_0)_\ast(x))g^{-1}=(f_\ast\circ\iota_\ast)(x)$ for any $x\in \pi_1(\partial M \setminus P, \ast_\partial)$, where $\iota_\ast$ is a homomorphism induced by the inclusion $\iota\colon \partial M \setminus P\to M \setminus \Gamma$ and a choice of path connecting $\ast$ and $\ast_\partial$.

\begin{figure}[htbp]
\centering\includegraphics[width=0.48\textwidth]{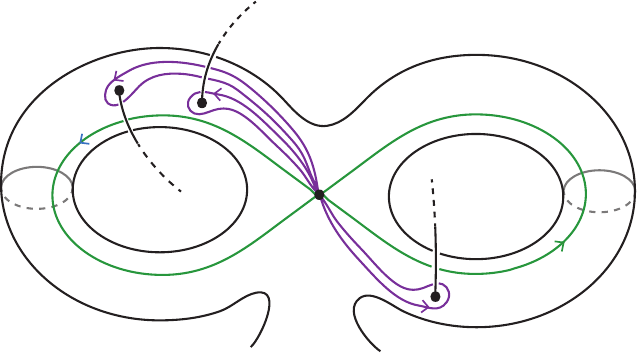}
\begin{picture}(400,0)(0,0)
\put(147,92){\color{violet} $\gamma_1$}
\put(120,95){\color{violet} $\gamma_2$}
\put(111,80){\color{teal} $\beta_1$}

\put(279,39){\color{teal} $\beta_g$}
\put(244,27){\color{violet} $\gamma_n$}

\put(193,10){$W$}
\end{picture}
\caption{The oriented based loops $\beta_1,\dots,\beta_{k}, \gamma_1,\dots,\gamma_n$ on the boundary of the handlebody $W$.}
\label{fig:loops_beta_gamma}
\end{figure}
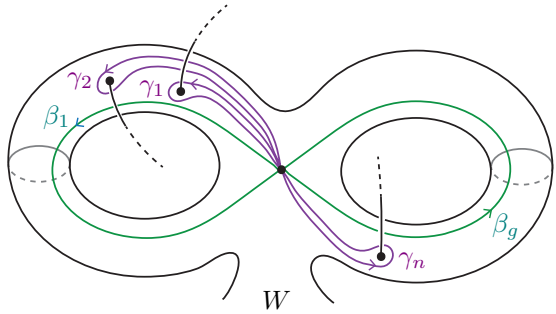

If $f\sim f'$ via $\{h_t\}_t$, then we have a commutative diagram
\[
\xymatrix{
\pi_1(M\setminus \widehat{\Gamma}) \ar@{->>}[r]^-{\incl_\ast} \ar[d]_-{h_1}^-{\cong} & \pi_1(M\setminus \Gamma) \ar[r]^-{f_\ast} & G \ar@{=}[d] \\
\pi_1(M\setminus \widehat{\Gamma}') \ar@{->>}[r]^-{\incl_\ast} & \pi_1(M\setminus \Gamma') \ar[r]^-{f'_\ast} &  G,
}
\]
and thus $\Im f_\ast=\Im f'_\ast$.
The commutative diagram also implies that $N_f=N_{f'}$ and $f_\ast\circ\iota_\ast = f'_\ast\circ\iota_\ast\colon \pi_1(\partial M \setminus P, \ast_\partial)\to G$.
Therefore, the map $\Phi$ is well-defined.

The main result of this paper is the following theorem, which gives the complete classification of defect configurations in handlebodies up to the above equivalence relation.
Inspired by the approach taken in \cite{NKTK24}, we will reduce our claim to a statement about colored diagrams, Theorem~\ref{thm:comb_classif} below.

\begin{theorem}
\label{thm:top_classif}
The map $\Phi$ is a bijection.
\end{theorem}

Note here that if $M$ is a $3$-ball and there are no boundary defects, then $N_f=\Im f_\ast$ and $C_G(\Im(f_0)_\ast)=G$ since $(f_0)_\ast$ is trivial, and thus $\SS_{G,\Im(f_0)_\ast,N_{\beta\gamma}}$ corresponds to $\SS^\fg_{G}$ in Section~\ref{subsec:previous_work}. 
Thus, \Cref{thm:top_classif} can naturally be thought of as a generalization of Theorem~\ref{thm:NKTK24}.

The next lemma enables us to describe continuous maps in terms of homomorphisms, 
which will be used in Section~\ref{sec:main_result}.
Note that the freedom to define $f|_{\partial M}$ %only 
up to homotopy through the full $X_G$ is essential here (cf.\ Remark~\ref{rem:boundary-conds}).

\begin{lemma}
\label{lem:MapHom}
The natural map \textup{(}see \Cref{subsec:homotopy}\textup{)}
\[
\theta\colon \Map_0(M \setminus \Gamma, X_G; f_0)/{\simeq}\to \{\varphi\in \Hom(\pi_1(M \setminus \Gamma,\ast),G)\mid \text{$\varphi\circ\iota_\ast$ is conjugate to $(f_0)_\ast$}\}
\]
is a bijection, where $\iota\colon \partial M \setminus P\hookrightarrow M \setminus \Gamma$ is the inclusion and $\simeq$ is the equivalence relation of basepoint-preserving homotopy.
%relative to $\partial H$.
\end{lemma}

\begin{proof}
First, note that we need to choose a path connecting $\ast$ and $\ast_\partial$ to define the induced homomorphism $\iota_\ast$, 
but the condition on the right-hand side does not depend on the choice of that path.
Let $K(G,1)$ be an Eilenberg-MacLane space obtained from $X_G$ by attaching cells of dimension greater than $3$ (see the construction in \cite[Section~4.2]{Hat02}).
Since $M\setminus \Gamma$ is homotopy equivalent to a CW complex of dimension $2$, the inclusion $X_G\subset K(G,1)$ induces a bijection $\Map_0(M \setminus \Gamma, X_G; f_0)/{\simeq} \to \Map_0(M \setminus \Gamma, K(G,1); f_0)/{\simeq}$.
By a property of Eilenberg-MacLane spaces, the natural map $\Map_0(M \setminus \Gamma, K(G,1))/{\simeq} \to \Hom(\pi_1(M \setminus \Gamma,\ast),G)$ is a bijection, and it restricts to the map $\theta$ on $\Map_0(M \setminus \Gamma, X_G; f_0)/{\simeq}$. 
This shows that $\theta$ is injective.

Now it suffices to show the surjectivity of $\theta$. Consider the following commutative diagram, where the horizontal arrows are bijections:
\[
\xymatrix{
\Map_0(M \setminus \Gamma, K(G,1))/{\simeq} \ar[r] \ar[d]^-{\iota^\ast} & \Hom(\pi_1(M \setminus \Gamma,\ast),G) \ar[d]^-{\iota^\ast} \\
\Map(\partial M \setminus P, K(G,1))/{\simeq} \ar[r]^-{\theta'} & \Hom(\pi_1(\partial M \setminus P,\ast_\partial),G)/\conj,
}
\]
Let $\varphi$ be an element of $\Hom(\pi_1(M \setminus \Gamma,\ast),G)$ such that $\varphi\circ\iota_\ast$ is conjugate to $(f_0)_\ast$, and let 
$f\colon M \setminus \Gamma \to K (G, 1)$ be the corresponding continuous map (up to homotopy).
Since the map $\theta'$ appearing in the diagram is a bijection, 
we have $[f\circ\iota] = \iota^\ast([f]) = [f_0]$, that is, the map $f\circ\iota$ is homotopic to $f_0$. 
\end{proof}

%%%
\subsection{Diagrammatic description}
\label{subsec:Diagrammatic_description}
For a positive integer $k$, a $k$-\emph{rose} is a graph consisting of a single 
vertex and $k$ loops.
For technical reasons, in this paper, we define a $0$-rose to be an interval with a basepoint.
We call a diagram of a spatial $k$-rose having no crossings 
the \emph{standard diagram} of a spatial $k$-rose.
Given the handlebody $M$ of genus $k$ as above, fix a deformation retraction of its exterior $\{h_t\colon W\to W\}_{t\in [0,1]}$ onto a $k$-rose $B$ such that $h_t|_{\partial M}$ is an embedding for $t\in [0,1)$, $h_1$ sends $\ast_\partial$ to the unique vertex (basepoint) of the rose, and $h_1|_{P\cup\{\ast_\partial\}}$ is injective.
We can identify $S^3\setminus \Int N(B)$ with $M = E(W)$, where $N(B)$ denotes a regular neighborhood of $B$.
When a spatial graph $\Gamma$ is given, we extend the edges incident to the boundary by $h_t$ until they reach $B$ and denote the resulting graph by $\Gamma$ again.
Then $S^3\setminus(B\cup \Gamma)$ is homeomorphic to $\Int M\setminus \Gamma$.
Consider a diagram of the graph $B \cup \Gamma$ such that the subdiagram corresponding 
to $B$ is standard, as on the right in Figure~\ref{fig:loops}.

Let $\Lambda=\{\alpha_1,\dots,\alpha_{k}, \beta_1,\dots,\beta_{k}, \gamma_1,\dots,\gamma_n\}$ be oriented loops on $\partial M$, based at $\ast_\partial$, such that their images under $h_\varepsilon$ are as illustrated in Figure~\ref{fig:loops} for $\varepsilon>0$ small enough. 
The loops $\alpha_j, \beta_k, \gamma_l$ form a generating set for $\pi_1(\partial M \setminus P, \ast_\partial)$. 
Note that, in the diagrammatic model illustrated on the right in Figure~\ref{fig:loops}, 
the loops $\alpha_1, \ldots, \alpha_g$ are not null-homotopic, although they may appear so---when we draw $B$, we really mean its $\varepsilon$-neighborhood $N(B)$ with basepoint on its boundary.
\begin{figure}[htbp]
\centering\includegraphics[width=1\textwidth]{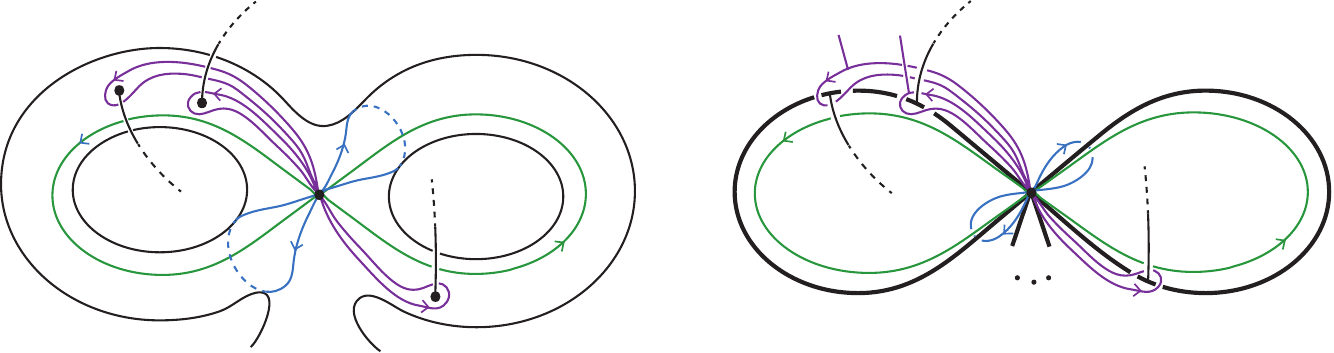}
\begin{picture}(400,0)(0,0)
\put(36,92){\color{violet} $\gamma_1$}
\put(8,95){\color{violet} $\gamma_2$}
\put(-1,80){\color{teal} $\beta_1$}
\put(83,40){\color{blue} $\alpha_1$}

\put(90,94){\color{blue} $\alpha_g$}
\put(168,39){\color{teal} $\beta_g$}
\put(132,27){\color{violet} $\gamma_n$}

\put(82,10){$W$}

\put(272,117){\color{violet} $\gamma_1$}
\put(247,117){\color{violet} $\gamma_2$}
\put(237,70){\color{teal} $\beta_1$}
\put(304,42){\color{blue} $\alpha_1$}

\put(322,85){\color{blue} $\alpha_g$}
\put(384,50){\color{teal} $\beta_g$}
\put(346,23){\color{violet} $\gamma_n$}

\put(255,15){$B$}
\end{picture}
\caption{(Left) The oriented based loops $\alpha_1,\dots,\alpha_{k}, \beta_1,\dots,\beta_{k}, \gamma_1,\dots,\gamma_n$ on the boundary of the exterior handlebody $W$.
(Right) Simplified version of the left figure, which will be used throughout the paper.
Here, $W$ is recovered from the bold rose $B$ by thickening.}
\label{fig:loops}
\end{figure}
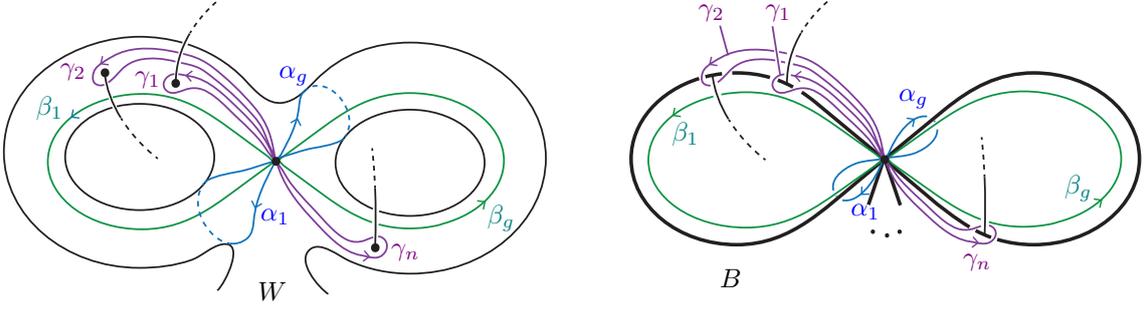

\begin{definition}
\label{def:colored_diagram}
A \emph{$G$-colored oriented spatial graph diagram relative to $f_0$} (or simply, \emph{$G$-colored diagram rel $f_0$}) is a diagram $D = D_B \cup D_\Gamma$ of an oriented spatial graph $B\cup \Gamma$ such that
\begin{enumerate}[label=(\arabic*)]
    \item each arc is colored by an element of $G$ (we write $\col(\alpha)$ for the color of an arc $\alpha\in \A_D$ of 
    $D$);
    \item at each crossing or vertex, the colors of the incident arcs satisfy the identities shown in Figure~\ref{fig:wirtinger2}, which come from the Wirtinger relations; and
    \begin{figure}[htbp]
\centering\includegraphics[width=0.70\textwidth]{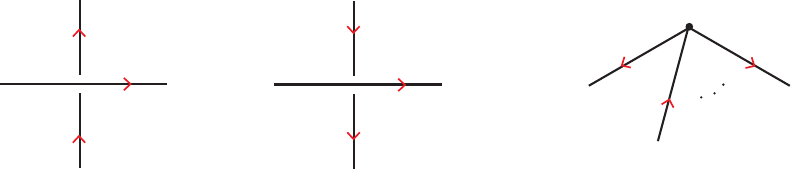}
\begin{picture}(400,0)(0,0)
\put(95,42){\color{red} $c_i$}
\put(66,70){\color{red} $c_j$}
\put(66,30){\color{red}
$c_k$}
\put(49,0){\color{red} $c_i c_j c_i^{-1} c_k^{-1} = 1$}

\put(198,42){\color{red} $c_i$}
\put(169,70){\color{red} $c_j$}
\put(169,30){\color{red} $c_k$}
\put(151,0){\color{red} $c_i c_j^{-1} c_i^{-1} c_k = 1$}

\put(273,66){\color{red} $c_{i_1}$}
\put(285,44){\color{red} $c_{i_2}$}
\put(333,66){\color{red} $c_{i_l}$}
\put(274,0){\color{red} $c_{i_1}^{\varepsilon_1}  c_{i_2}^{\varepsilon_2} \cdots c_{i_l}^{\varepsilon_l} = 1$}
\end{picture}
\caption{Colors around crossings and vertices.}
\label{fig:wirtinger2}
\end{figure}
    \item there exists $g\in G$ such that, for each loop $\lambda \in \Lambda$, $g[f_0( \lambda )]g^{-1}=\prod_{\alpha} \col(\alpha)^{\varepsilon_\alpha} \in G$, where $\alpha$ runs over all arcs crossing over $\lambda$ and $\varepsilon_\alpha=1$ (resp.\ $\varepsilon_\alpha=-1$) if the crossing is positive (resp.\ negative). \label{item:basis-loop-identity}
\end{enumerate}
\end{definition}

Note that $g$ in \ref{item:basis-loop-identity} is unique as an element of $G/C_G(\Im(f_0)_\ast)$. 
The left quotient appearing in the definition of $\SS_{G,\Im(f_0)_\ast,N_{\beta\gamma}}$ comes from the equivalence between diagrams.

\begin{example}
In \Cref{fig:loops_ex}, the right-hand side of the identity in \ref{item:basis-loop-identity} for $\lambda=\gamma_2$ is $c_9^{-1}c_7c_9$.
\begin{figure}[htbp]
\centering\includegraphics[width=0.98\textwidth]{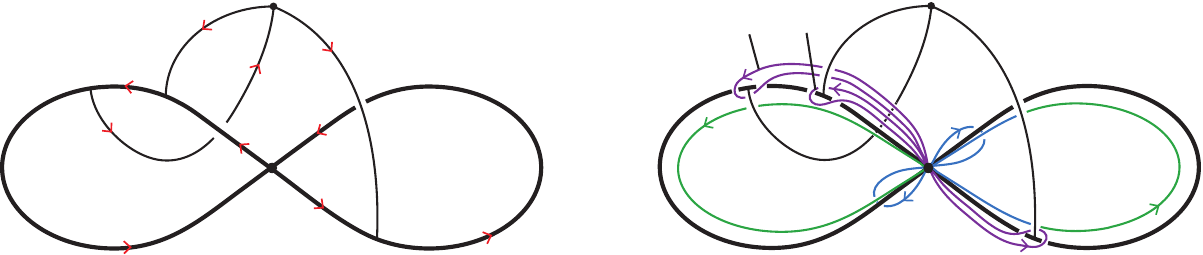}
\begin{picture}(400,0)(0,0)
\put(69,41){\color{red} $c_1$}
\put(33,76){\color{red} $c_2$}
\put(30,5){\color{red} $c_3$}
\put(95,20){\color{red} $c_4$}
\put(158,8){\color{red} $c_5$}
\put(95,60){\color{red} $c_6$}
\put(18,50){\color{red} $c_7$}
\put(86,73){\color{red} $c_8$}
\put(53,95){\color{red} $c_9$}
\put(110,85){\color{red} $c_{10}$}

\put(268,94){\color{violet} $\gamma_1$}
\put(245,94){\color{violet} $\gamma_2$}
\put(237,48){\color{teal} $\beta_1$}
\put(303,22){\color{blue} $\alpha_1$}
\put(320,61){\color{blue} $\alpha_2$}
\put(384,33){\color{teal} $\beta_2$}
\put(340,6){\color{violet} $\gamma_3$}

\end{picture}
\caption{An example of a $G$-colored diagram for $k=2$, $n=3$, and the oriented loops.}
\label{fig:loops_ex}
\end{figure}
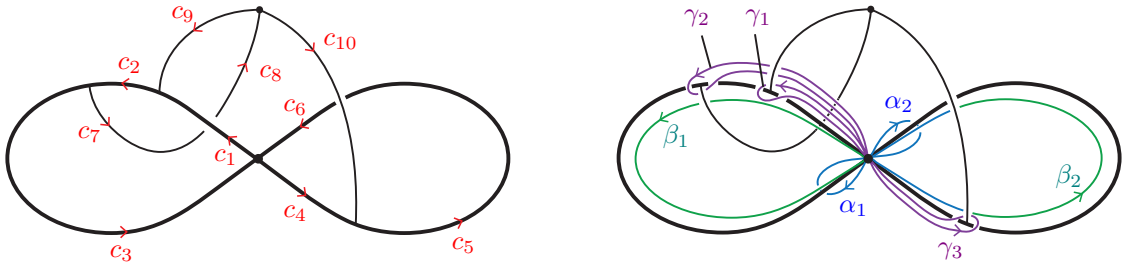
\end{example}

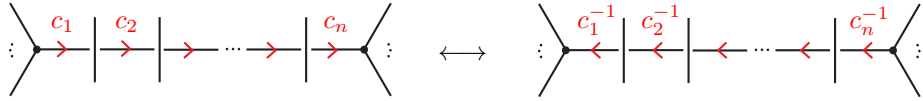
\begin{figure}[htbp]
\centering\includegraphics[width=12cm]{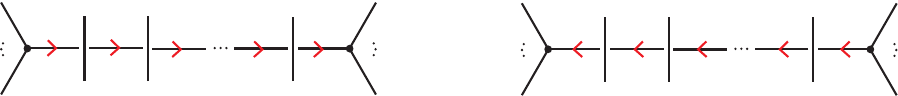}
\begin{picture}(400,0)(0,0)

\put(45,38){\color{red} $c_1$}
\put(69,38){\color{red} $c_2$}
\put(147,38){\color{red} $c_n$}

\put(242,38){\color{red} $c_1^{-1}$}
\put(266,38){\color{red} $c_2^{-1}$}
\put(344,38){\color{red} $c_n^{-1}$}

\put(190,26){$\longleftrightarrow$}
\end{picture}
\caption{Orientation-reversal.}
\label{fig:local_move_1}
\end{figure}

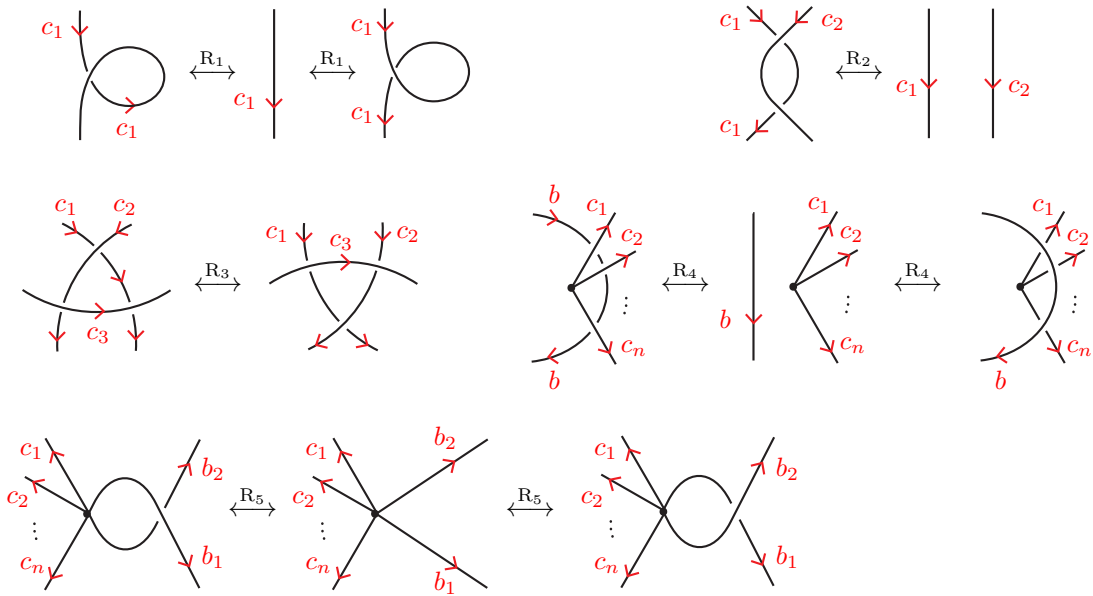
\begin{figure}[htbp]
\centering\includegraphics[width=14cm]{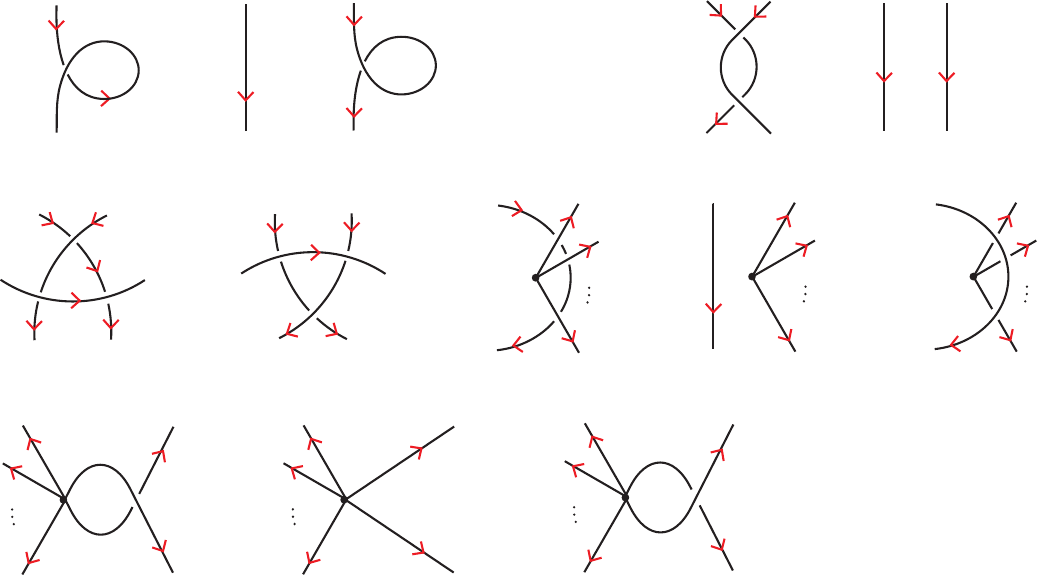}
\begin{picture}(400,0)(0,0)

\put(8,222){\color{red} $c_1$}
\put(36,184){\color{red} $c_1$}
\put(80,195){\color{red} $c_1$}
\put(124,190){\color{red} $c_1$}
\put(124,223){\color{red} $c_1$}
\put(63,205){$\overset{\RR_1}{\longleftrightarrow}$}
\put(108,205){$\overset{\RR_1}{\longleftrightarrow}$}

\put(262,225){\color{red} $c_1$}
\put(300,225){\color{red} $c_2$}
\put(262,185){\color{red} $c_1$}
\put(305,205){$\overset{\RR_2}{\longleftrightarrow}$}
\put(328,200){\color{red} $c_1$}
\put(370,200){\color{red} $c_2$}

\put(13,155){\color{red} $c_1$}
\put(35,155){\color{red} $c_2$}
\put(25,107){\color{red} $c_3$}
\put(65,125){$\overset{\RR_3}{\longleftrightarrow}$}
\put(92,145){\color{red} $c_1$}
\put(141,145){\color{red} $c_2$}
\put(116,142){\color{red} $c_3$}

\put(198,158){\color{red} $b$}
\put(198,88){\color{red} $b$}
\put(212,155){\color{red} $c_1$}
\put(225,144){\color{red} $c_2$}
\put(225,102){\color{red} $c_n$}
\put(240,125){$\overset{\RR_4}{\longleftrightarrow}$}
\put(262,111){\color{red} $b$}
\put(295,156){\color{red} $c_1$}
\put(307,145){\color{red} $c_2$}
\put(307,103){\color{red} $c_n$}
\put(327,125){$\overset{\RR_4}{\longleftrightarrow}$}
\put(365,88){\color{red} $b$}
\put(379,155){\color{red} $c_1$}
\put(392,144){\color{red} $c_2$}
\put(392,102){\color{red} $c_n$}

\put(0,63){\color{red} $c_1$}
\put(-5,45){\color{red} $c_2$}
\put(0,20){\color{red} $c_n$}
\put(68,55){\color{red} $b_2$}
\put(68,22){\color{red} $b_1$}
\put(78,40){$\overset{\RR_5}{\longleftrightarrow}$}
\put(107,63){\color{red} $c_1$}
\put(102,45){\color{red} $c_2$}
\put(107,20){\color{red} $c_n$}
\put(155,67){\color{red} $b_2$}
\put(155,12){\color{red} $b_1$}
\put(182,40){$\overset{\RR_5}{\longleftrightarrow}$}
\put(215,63){\color{red} $c_1$}
\put(210,45){\color{red} $c_2$}
\put(215,20){\color{red} $c_n$}
\put(283,55){\color{red} $b_2$}
\put(283,22){\color{red} $b_1$}
\end{picture}
\caption{Reidemeister moves: The colors of the short `local' arcs are omitted because they are uniquely determined by the colors of the other arcs, according to the 
appropriate Wirtinger relators.
}
\label{fig:local_move_2}
\end{figure}
\begin{figure}[htbp]
    \centering
    \includegraphics[width=1\textwidth]{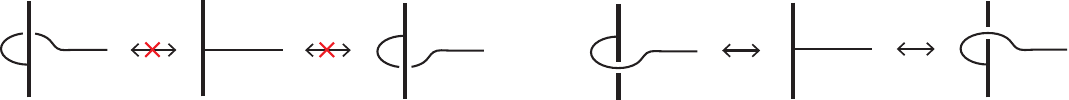}
    \caption{The forbidden $\RR_5$ moves (left) and examples of (allowed) $\RR_5$ moves (right), where vertical bold lines are parts of $B$.}
    \label{fig:forbidden_R5}
\end{figure}

We here introduce local moves among $G$-colored diagrams $D = D_B \cup D_\Gamma$. 
The diagram $D_B$ of the rose $B$ is assumed to be fixed during those moves. 
See \cite{NKTK24} for more detailed description. 

\begin{definition}
\label{def:moves}
\begin{enumerate}[label=(\arabic*)]
    \item An \emph{orientation-reversal} for $D$, 
    which is described in \Cref{fig:local_move_1}. 

    \item \emph{Reidemeister moves} for $D$, which are described in \Cref{fig:local_move_2}, except for the forbidden $\RR_5$ moves shown in \Cref{fig:forbidden_R5}.

    \item An \emph{edge-contraction} or a \emph{vertex-splitting} for $D_\Gamma$.
    As special cases, we obtain an \emph{edge-contraction} or an \emph{edge-subdivision}.
    See Figure~\ref{fig:local_moves_3-4}. 
    
    \item An \emph{edge-addition} or an \emph{edge-deletion} for $D_\Gamma$, which are described in 
    Figure~\ref{fig:local_move_5}. 

    \item A \emph{vertex-addition} or a \emph{vertex-deletion} for $D_\Gamma$, which adds or deltes an isolated vertex. 
    
    \item A \emph{simultaneous conjugation} for $D$, 
    which replaces all colors $c_1, c_2,\ldots, c_n$ with their conjugates $x c_1 x^{-1}, x c_2 x^{-1}, \ldots, x c_n x^{-1}$ 
    by a fixed element $x \in G$ simultaneously. 
\end{enumerate}
\end{definition}

\begin{figure}[h]
\centering\includegraphics[width=15cm]{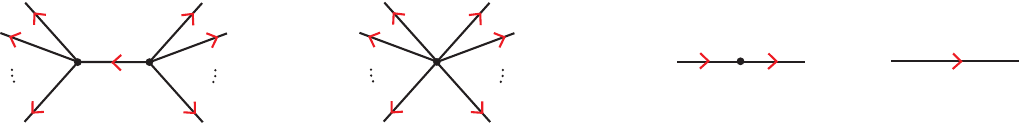}
\begin{picture}(400,0)(0,0)
\put(277,43){\color{red} $c$}
\put(306,43){\color{red} $c$}
\put(383,43){\color{red} $c$}
\put(331,35){$\longleftrightarrow$}

\put(2,62){\color{red} $b_1$}
\put(-18,39){\color{red} $b_2$}
\put(-14,18){\color{red} $b_m$}

\put(19,27){\color{red} $b_1 \cdots b_m$}

\put(56,64){\color{red} $c_1$}
\put(76,39){\color{red} $c_2$}
\put(71,20){\color{red} $c_n$}

\put(97,35){$\longleftrightarrow$}

\put(155,62){\color{red} $b_1$}
\put(135,39){\color{red} $b_2$}
\put(137,18){\color{red} $b_m$}
\put(178,62){\color{red} $c_1$}
\put(197,39){\color{red} $c_2$}
\put(193,18){\color{red} $c_n$}
\end{picture}
\vspace{-1em}
\caption{An edge-contraction and a vertex-splitting (left); 
an edge-combining and an edge-subdivision (right).}
\label{fig:local_moves_3-4}
\end{figure}
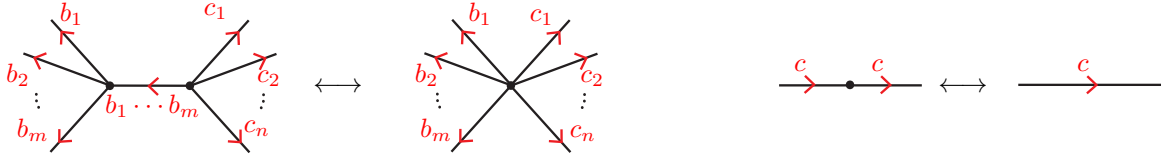

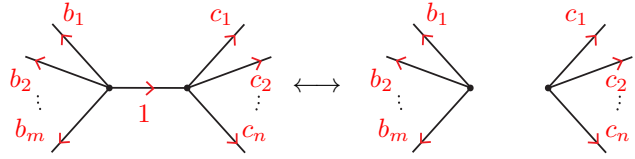
\begin{figure}[h]
\centering\includegraphics[width=8cm]{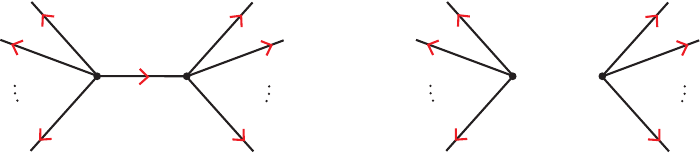}
\begin{picture}(400,0)(0,0)
\put(102,62){\color{red} $b_1$}
\put(82,37){\color{red} $b_2$}
\put(84,18){\color{red} $b_m$}
\put(157,62){\color{red} $c_1$}
\put(172,37){\color{red} $c_2$}
\put(169,18){\color{red} $c_n$}

\put(188,35){$\longleftrightarrow$}

\put(238,62){\color{red} $b_1$}
\put(217,37){\color{red} $b_2$}
\put(220,18){\color{red} $b_m$}
\put(130,24){\color{red} $1$}
\put(290,62){\color{red} $c_1$}
\put(305,37){\color{red} $c_2$}
\put(305,18){\color{red} $c_n$}
\end{picture}
\vspace{-1em}
\caption{Edge-addition and edge-deletion.}
\label{fig:local_move_5}
\end{figure}

%%%
\section{Main results and examples}
\label{sec:main_result}
For a $G$-colored diagram $D$ rel $f_0$, let $H_D$ denote the subgroup $\ang{\col(\alpha)\mid \alpha\in \A_{D_B\cup D_\Gamma}}$ of $G$ and let $N_D$ denote the normal closure of $\{\col(\alpha)\mid \alpha\in \A_{D_\Gamma}\}$ in $H_D$.
One can define a map
\[
\Psi\colon \{\text{$G$-colored diagrams rel $f_0$}\}/\text{(1)--(5)} \to \SS_{G,\Im(f_0)_\ast,N_{\beta\gamma}}
\]
by $\Psi(D)=(H_D, N_D, [g])$, where $g$ is an element appearing in \Cref{def:colored_diagram}\ref{item:basis-loop-identity}.
Then, $\Psi$ is well defined since $H_D$, $N_D$, and $[g]$ are invariant under moves (1)--(5).
Note that move (2) changes $g$ to $hg$ for some $h\in N_D$ when an arc goes over the vertex of $D_B$, but these two elements are the same in $N_D \backslash G/C_G(\Im(f_0)_\ast)$.

The map $\Psi$ fits into a commutative diagram
\[
\xymatrix{
\{\text{$G$-colored diagrams rel $f_0$}\}/\text{(1)--(5)} \ar[r]^-{\Psi} \ar@{->>}[d] & \SS_{G,\Im(f_0)_\ast,N_{\beta\gamma}}, \\
\Bigl(\coprod_\Gamma \Map_0( M \setminus \Gamma, X_G; f_0)\Bigr)/{\sim} \ar[ur]_-{\Phi} &
}
\]
where the vertical map is induced by Lemma~\ref{lem:MapHom}.
Indeed, for a $G$-colored diagram $D$, we obtain an element of $\{\varphi\in \Hom(\pi_1( M \setminus \Gamma,\ast),G)\mid \text{$(f_0)_\ast$ is conjugate to $\varphi\circ\iota_\ast$}\}$ due to the condition on colors in Definition~\ref{def:colored_diagram}, then Lemma~\ref{lem:MapHom} gives an element of $\Map_0( M \setminus \Gamma, X_G; f_0)/{\simeq}$.

\begin{theorem}
\label{thm:comb_classif}
The map $\Psi$ is a bijection.
\end{theorem}

Now, Theorem~\ref{thm:top_classif} follows from Theorem~\ref{thm:comb_classif} and the above commutative diagram.
Indeed, Theorem~\ref{thm:comb_classif} implies that the vertical map is a bijection, and hence $\Phi$ is also a bijection.

If we ignore basepoints, the above commutative diagram implies the following consequence.

\begin{corollary}
\label{cor:comp_classif_conjugacy}
The diagram
\[
\xymatrix{
\{\text{$G$-colored diagrams rel $f_0$}\}/\textup{(1)--(6)} \ar[r]^-{\Psi} \ar@{->>}[d] & \SS_{G,\Im(f_0)_\ast,N_{\beta\gamma}}/\conj \\
\Bigl(\coprod_\Gamma \Map(M\setminus \Gamma, X_G; f_0)\Bigr)/{\sim} \ar[ur]_-{\Phi} &
}
\]
is commutative, where the three maps are bijections.
Here $\SS_{G,\Im(f_0)_\ast,N_{\beta\gamma}}/\conj$ is the quotient set obtained by identifying $(H,N,[g])$ and $(H',N',[g'])$ if $(xHx^{-1},xNx^{-1},[xg])=(H',N',[g'])$ for some $x\in G$.
\end{corollary}

Note that we have a bijection between $\SS_{G,\Im(f_0)_\ast,N_{\beta\gamma}}/\conj$ and the set of pairs $(H,N)$ up to conjugation by an element $yz$, where $y\in N$ and $z\in C_G(\Im(f_0)_\ast)$.

\begin{remark}
\label{rem:constant far-field}
Our classification of defects up to basepoint-preserving homotopy of ordered media in $M$ also covers the case of defects in ordered media with constant far-field on the 3-dimensional space $\R^3$ outside of an unknotted handlebody (through one-point compactification of $\R^3$). 
The corollary covers the case of free homotopy of ordered media in the interior of an (unknotted) handlebody.
\end{remark}

\begin{example}
\label{ex:meaning_of_g}
Let $X_G = S^3 /Q$, genus $k = 1$, and $P = \{ p_1, p_2\}$. 
%Let us consider the case where $X_G=S^3/Q$, the genus of a handlebody is one, 
%and $n=2$.
Fix a continuous map $f_0\colon \partial M\setminus\{p_1,p_2\}\to S^3/Q$ satisfying $(f_0)_\ast(\alpha_1)=(f_0)_\ast(\beta_1)=1$, $(f_0)_\ast(\gamma_1)=i$, and $(f_0)_\ast(\gamma_2)=-i$.
Then, $\Im(f_0)_\ast = N_{\beta\gamma} = \ang{i}$, and thus $H=N$ for any $(H,N,[g])\in \SS_{Q,\ang{i},\ang{i}}$.
It follows that
\[
\SS_{Q,\ang{i},\ang{i}} = \{(\ang{i},\ang{i},[1]),\ (\ang{i},\ang{i},[j]),\ (Q,Q,[1])\}.
\]
Moreover, we have
\[
\SS_{Q,\ang{i},\ang{i}}/\conj = \{(\ang{i},\ang{i},[1]),\ (Q,Q,[1])\}.
\]
\Cref{fig:example_12} shows a diagrammatic realization of each element of $\SS_{Q,\ang{i},\ang{i}}$.
The configurations shown in the left and middle diagrams differ only by $[g]$. The defect configuration in the middle diagram can be obtained from that in the left diagram by dragging the image under $f$ of a neighborhood of the basepoint $\ast$ along the loop in $S^3/Q$ corresponding to $j$.
\end{example}
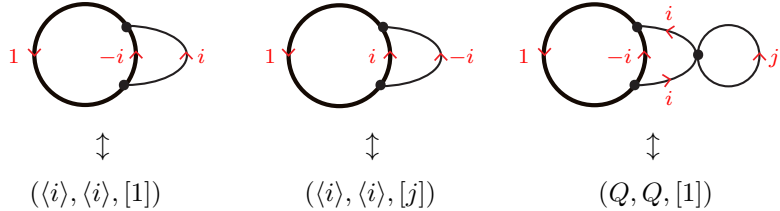
\begin{figure}[h]
    \centering
    \includegraphics[width=0.65\textwidth]{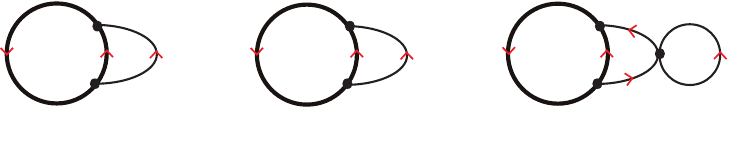}
\begin{picture}(400,0)(0,0)

\put(88,52){\footnotesize \color{red}$-i$}
\put(125,52){\footnotesize \color{red}$i$}
\put(189,52){\footnotesize \color{red}$i$}
\put(219,52){\footnotesize \color{red}$-i$}
\put(278,52){\footnotesize \color{red}$-i$}
\put(300,68){\footnotesize \color{red}$i$}
\put(300,36){\footnotesize \color{red}$i$}
\put(339,52){\footnotesize \color{red}$j$}

\put(54,52){\footnotesize \color{red}$1$}
\put(150,52){\footnotesize \color{red}$1$}
\put(245,52){\footnotesize \color{red}$1$}

\put(62,0){$(\langle  i \rangle, \langle  i \rangle, [1])$}
\put(165,0){$(\langle  i \rangle , \langle i \rangle, [j])$}
\put(275,0){$(Q, Q , [1])$}

\put(87,18){$\updownarrow$}
\put(188,18){$\updownarrow$}
\put(293,18){$\updownarrow$}
\end{picture}
\caption{Three non-equivalent $Q$-colored diagrams rel $f_0$.}
    \label{fig:example_12}
\end{figure}

In \Cref{sec: Classification example}, 
we give a complete classification, 
up to free homotopy, of non-trivial defects of a system defined on the solid torus $M = D^2 \times S^1$ that has 
$S^3 / Q$ as its order parameter space and no boundary defects such that 
the boundary condition is given by $f_0 : \partial M \to S^3 / Q$ with $(f_0)_* (\alpha) = i$. 

%%%
\section{Proof of Theorem~\ref{thm:comb_classif}}
\label{sec:proof_of_theorem}

Recall that $B$ is a $k$-rose in $W$ with a fixed 
deformation retraction $\{h_t\colon W\to W\}_{t\in [0,1]}$ onto $B$ such that $h_t|_{\partial W}$ is an embedding for $t\in [0,1)$, $h_1$ sends $\ast_\partial$ to the unique vertex 
of the rose, and $h_1|_{P\cup\{\ast_\partial\}}$ is injective.
Consider a diagram $D = D_B \cup D_\Gamma$ of the graph $B \cup \Gamma$ such that the the subdiagram $D_B$ corresponding 
to $B$ is standard. 

We say that a $G$-colored diagram $D$ is in a \emph{standard form} if it satisfies the following conditions
(see Figure~\ref{fig:standard_form}):
\begin{figure}[h]
    \centering
    \includegraphics[width=0.50\textwidth]{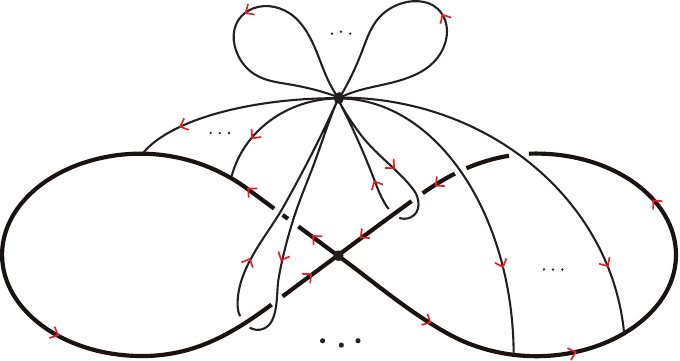}
    \caption{A Standard form of $G$-colored diagram, where the colors on edges are omitted.}
    \label{fig:standard_form}
\end{figure}

\begin{itemize}
\item 
$D_\Gamma$ has a single vertex $v_0$ that does not touch $D_B$; 
\item 
$D_\Gamma$ has no self-crossings; 
\item
for each edge of $D_B$, there exists a unique loop of $D_\Gamma$ that is hooked to that edge, that is, the corresponding loop of $\Gamma$ is freely homotopic to the meridian of the edge of $B$; 
\item
for each point of $D_B$ corresponding to one of the boundary defects $p_1, \ldots, p_n$, there is a unique edge of $D_\Gamma$ connecting $v_0$ and that point; and
\item
at every crossing of $D_B$ and $D_\Gamma$, except the hooks mentioned above, $D_\Gamma$ crosses over $D_B$. 
\end{itemize}
For a diagram $D = D_B \cup D_\Gamma$ in a standard form, 
we divide the edges of $D_\Gamma$ into three as follows: 
\begin{itemize}
\item the set $D_\Gamma^{(1)}$ of loops of $D_\Gamma$ that are not 
hooked to an edge of $D_B$;  
\item the set $D_\Gamma^{(2)}$ of loops of $D_\Gamma$ that are hooked to an edge of $D_B$; and 
\item the set $D_\Gamma^{(3)}$ of edges of $D_\Gamma$ that are not loops. 
\end{itemize}

\begin{lemma}
\label{lem:colors of a diagram coming from the boudnary}
Let $D = D_B \cup D_\Gamma$ be a $G$-colored diagram 
rel $f_0$ in a standard form. 
Then the colors of the arcs contained in 
the edges of $D_\Gamma^{(2)} \cup D_\Gamma^{(3)}$ 
are completely determined by $f_0$ 
and $g$ defined in the condition $(3)$ in Definition~\ref{def:colored_diagram}.  
\end{lemma}

\begin{proof}
The conditions (2) and (3) in \Cref{def:colored_diagram} provide equations in $G$ relating the arc colors of $D_\Gamma^{(2)} \cup D_\Gamma^{(3)} \cup D_B$. The idea is to proceed inductively from the basepoint of $D_B$, solving these equations for the arc colors.

Consider, for example, a part of a $G$-colored diagram 
in a standard form shown in Figure~\ref{fig:standard_form_coloring}. 
\begin{figure}[h]
    \centering
    \includegraphics[width=0.8\textwidth]{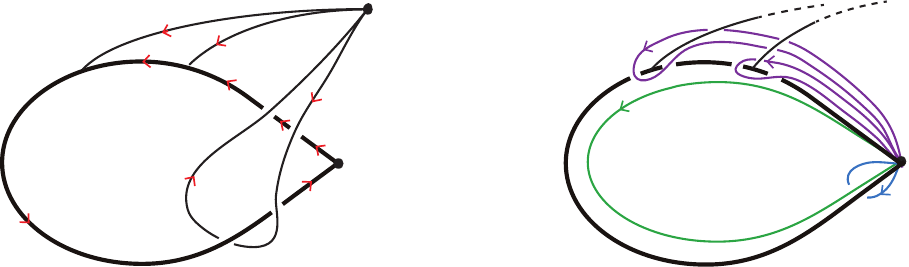}
\begin{picture}(400,0)(0,0)
\put(146,34){\color{red} $c_1$}
\put(29,23){\color{red} $c_2$}
\put(80,80){\color{red} $c_3$}
\put(108,72){\color{red} $c_4$}
\put(127,58){\color{red} $c_5$}
\put(148,62){\color{red} $c_6$}

\put(152,72){\color{red} $d_1$}
\put(88,40){\color{red} $d_2$}
\put(114,89){\color{red} $d_3$}
\put(85, 105){\color{red} $d_4$}

\put(355,30){\color{blue} $\alpha_1$}
\put(260,60){\color{teal} $\beta_1$}
\put(295,91){\color{violet} $\gamma_1$}
\put(260, 95){\color{violet} $\gamma_2$}
\end{picture}
    \caption{A part of a $G$-colored diagram in a 
    standard form.}
    \label{fig:standard_form_coloring}
\end{figure}
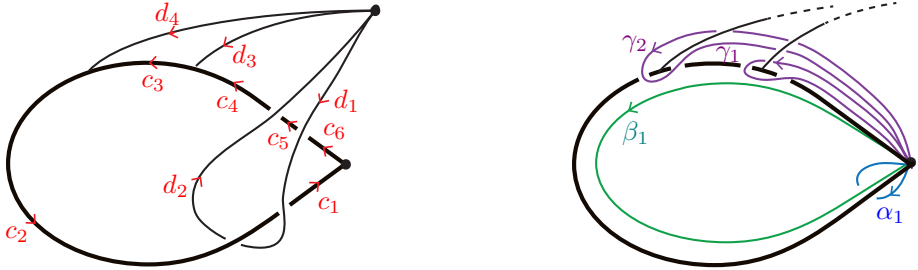
Then we have $g f_0(\alpha_1) g^{-1} = c_1$, which implies that $c_1$ is determined by $f_0$ and $g$.
Since $c_2 = d_1 c_1 d_1^{-1}$, we can write 
$d_2 = c_2^{-1} d_1 c_2 = d_1 c_1^{-1} d_1 c_1 d_1^{-1}$. 
Thus, we have $g [ f_0(\beta_1) ] g^{-1} = d_1^{-1} d_2 d_1 = c_1^{-1} d_1 c_1$, which implies that $d_1$ is determined by $f_0$ and $g$. 
The equality $g [ f_0(\beta_1) ] g^{-1} = d_1^{-1} d_2 d_1$ then implies that $d_2$ is also determined by $f_0$. 
Since we have $g [ f_0(\gamma_1) ] g^{-1} = d_1^{-1} d_2 d_3^{-1} d_2^{-1}d_1$, the color $d_3$ is determined by $f_0$ and $g$. 
Similarly, since $g [ f_0(\gamma_2) ] g^{-1} = 
d_1^{-1} d_2 d_3^{-1} d_4^{-1} d_3 d_2^{-1}d_1$, 
the color $d_4$ is also determined only by $f_0$ and $g$.
Now, the remaining colors $c_2, \ldots, c_6$ are determined 
from the other colors by the condition (2) in Definition~\ref{def:colored_diagram}. 
Consequently, all colors $c_1 \ldots, c_6, d_1, \ldots, d_4$ are determined by $f_0$. 
The general case can be proved in the same way with additional indices. 
\end{proof}

\begin{proof}[Proof of Theorem~\ref{thm:comb_classif}]
We first show the surjectivity of $\Psi$. 
Let $(H,N,[g]) \in \SS_{G,\Im(f_0)_\ast,N_{\beta\gamma}}$, that is, $[g] \in N\backslash G/C_G(\Im(f_0)_\ast)$, $g \Im(f_0)_\ast g^{-1}\subset H$, $g N_{\beta\gamma} g^{-1}\subset N$ and there exists a finite subset $T=\{c_1,\dots,c_m\}$ of $N$ satisfying $\ang{g \Im(f_0)_\ast g^{-1}\cup T}=H$ and $\angg{T}_H=N$.
We consider a (non-colored) diagram $D = D_B \cup D_\Gamma$ in a standard form such that 
the set $D_\Gamma^{(1)}$ consists of $n$ loops. 
Color the loops of $D_\Gamma^{(1)}$ by $\{c_1, \ldots, c_m\}$, respectively. 
By the arguments of Lemma~\ref{lem:colors of a diagram coming from the boudnary}, 
we can define the colors of the edges of 
$D_\Gamma^{(2)} \cup D_\Gamma^{(3)}$ automatically 
from $f_0$ and $g$. 
Then, the resulting $G$-colored diagram $D$ satisfies $\ang{\col(\alpha)\mid \alpha\in \A_{D}}=H$ and $\angg{\{\col(\alpha)\mid \alpha\in \A_{D_\Gamma}\}\cup T}_{H}=N$.
Thus $D$ lies in the preimage $\Psi^{-1}(H,N,[g])$.

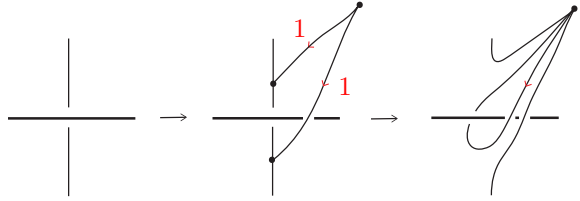
\begin{figure}[htbp]
\centering\includegraphics[width=0.5\textwidth]{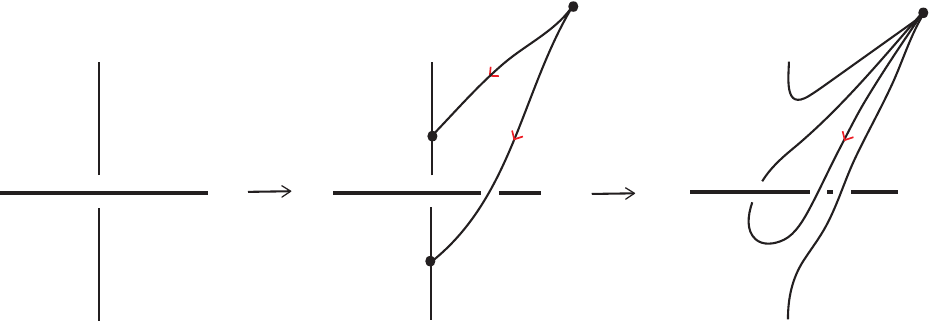}
\begin{picture}(400,0)(0,0)
\put(200,72){\color{red} $1$}
\put(217,50){\color{red} $1$}
\end{picture}
\caption{Adding two edges colored with $1$ from the unique vertex to arcs crossing under $D_B$, and then contracting them.}
\label{fig:hook1}
\end{figure}
\begin{figure}[htbp]
\centering\includegraphics[width=0.75\textwidth]{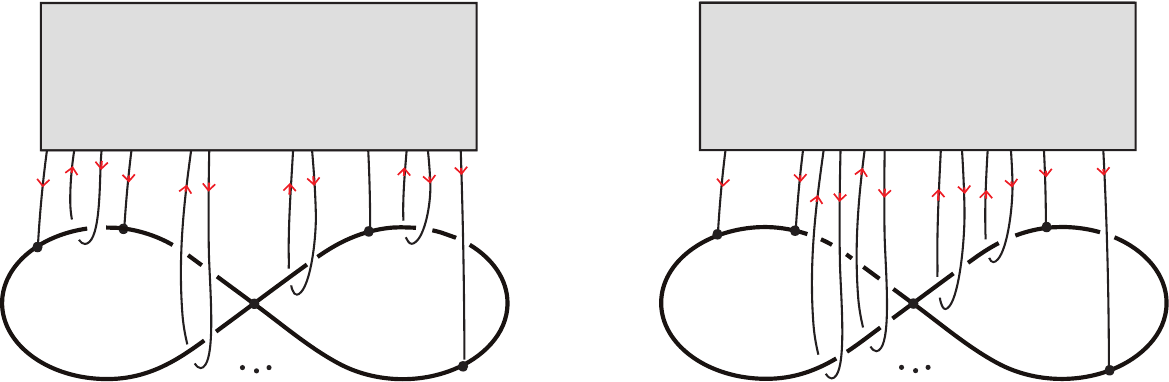}
\begin{picture}(400,0)(0,0)
%\put(69,41){\color{red} $c_1$}
\end{picture}
\caption{Simplification of colored graphs in the neighborhood of the $k$-rose $B$.}
\label{fig:hook2}
\end{figure}

It remains to show that $\Psi$ is injective. 
Choose $(H,N,[g]) \in \SS_{G,\Im(f_0)_\ast,N_{\beta\gamma}}$.
Let $D = D_B \cup D_\Gamma$ be an arbitrary $G$-colored diagram rel $f_0$ in the preimage $\Psi^{-1}(H,N,[g])$. 
We show that by a sequence of local moves (1)--(5) for 
$D = D_B \cup D_\Gamma$ introduced in 
Section~\ref{sec:Mathematical model}, 
$D$ is transformed into a standard form. 
First, by adding edges colored by $1$ appropriately, we 
can assume that a graph $\Gamma$ corresponding to $D_\Gamma$ is connected. 
Then, by contracting the edges of a maximal tree of $\Gamma$, 
we can assume that $\Gamma$ has a single vertex, where the image of the vertex is on the top of the diagram. 
Furthermore, by applying the operations depicted in Figure~\ref{fig:hook1} to each crossing where a strand of $D_\Gamma$ passes under a strand of $D_B$, we obtain a diagram like that shown on the left of Figure~\ref{fig:hook2}, or more precisely, a diagram satisfying the following conditions: 
\begin{itemize}
\item 
outside of the shaded box, $D_\Gamma$ has no self-crossings; 
\item
for each point of $D_B$ corresponding to one of the boundary defects $p_1, \ldots, p_n$, there is a portion of an edge of $D_\Gamma$ extending from the shaded box to that point; 
\item
each portion of $D_\Gamma$'s edges outside of the shaded box, apart from those arcs described above, is hooked to an edge of $D_B$ so that 
if we contract the shaded box to a point, that portion represents a loop that is freely homotopic to the meridian of the corresponding edge of $D_B$; and 
\item
at every crossing of $D_B$ and $D_\Gamma$, except for the hooks mentioned above, $D_\Gamma$ crosses over $D_B$.
\end{itemize}

\begin{figure}[htbp]
\centering\includegraphics[width=0.9\textwidth]{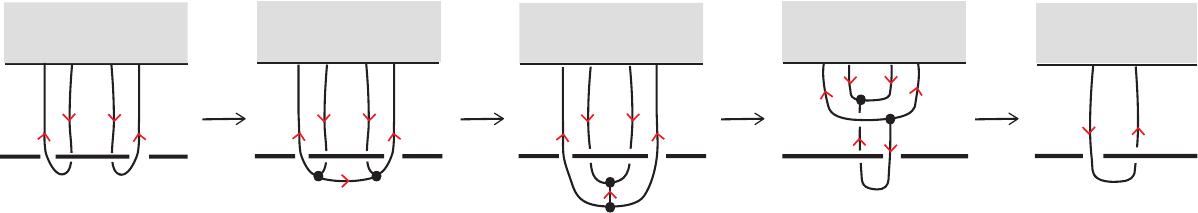}
\begin{picture}(400,0)(0,0)
\end{picture}
\caption{A sequence of moves that reduces the number of loops hooked to an edge of $D_B$.}
\label{fig:hook3}
\end{figure}

By the moves (1)--(5) in \Cref{def:moves}, the images of the loops of $\Gamma$ hooked to each edge of $D_B$ can be collected in one place as shown on the right in Figure~\ref{fig:hook2}. 
Then, after applying the moves shown in Figure~\ref{fig:hook3} finitely many times, we may assume that for each edge of $D_B$, there exists a unique loop of $D_\Gamma$ that is hooked to that edge. 

Now, the diagram is in a standard form outside of the grayed-out box. 
The remaining argument to get a standard form is similar to that of
\textcite{NKTK24}. 
Indeed, the move shown in Figure~\ref{fig:crossing_contraction} replaces a crossing with a vertex; thus, we can assume that the diagram $D_\Gamma$ has no self-crossings. 
By contracting edges to get a diagram with a single vertex, 
we finally obtain a diagram in a standard form. 

\begin{figure}[htbp]
\centering\includegraphics[width=0.8\textwidth]{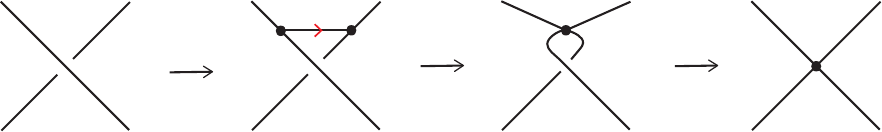}
\begin{picture}(400,0)(0,0)
\end{picture}
\caption{A sequence of moves changing a crossing to a vertex.}
\label{fig:crossing_contraction}
\end{figure}

Let $D = D_B \cup D_\Gamma$ and $D'= D_B \cup D_{\Gamma'}$ 
be $G$-colored diagrams in standard form satisfying $\Psi(D) = \Psi(D') = (H,N,[g])$.
Let $g$ and $g'$ be elements of $G$ defined in the condition (3) in Definition~\ref{def:colored_diagram} for $D$ and $D'$, respectively.
Then, $hgs=g'$ for some $h\in N$ and $s\in C_G(\Im(f_0)_\ast)$.
Figure~\ref{fig:conjugate} enables us to create a loop colored by $h$.
Following Figure~\ref{fig:returning}, we can transform $D'$ to another standard form $D''$ such that $g''=h^{-1}g'=gs$.

Clearly, the number of edges of 
$D_\Gamma^{(2)}$ and $D_{\Gamma'}^{(2)}$ as well as 
$D_\Gamma^{(3)}$ and $D_{\Gamma'}^{(3)}$ are the same. 
Further, by Lemma~\ref{lem:colors of a diagram coming from the boudnary}, the colors of the arcs in 
$D_\Gamma^{(2)} \cup D_{\Gamma'}^{(2)}$ and 
$D_\Gamma^{(3)} \cup D_{\Gamma'}^{(3)}$ are automatically 
determined by $f_0, g$ and $f_0, gs$, respectively. 
Since $s\in C_G(\Im(f_0)_\ast)$, it does not affect this process. 
Finally, since $N_D=N=N_{D'}$, we can deform $D_{\Gamma'}^{(1)}$ into $D_\Gamma^{(1)}$ using \cite[Figure~16]{NKTK24}.
Therefore, $D=D'$ up to the moves (1)--(5).
\end{proof}

\begin{figure}[h]
\centering\includegraphics[width=0.9\textwidth]{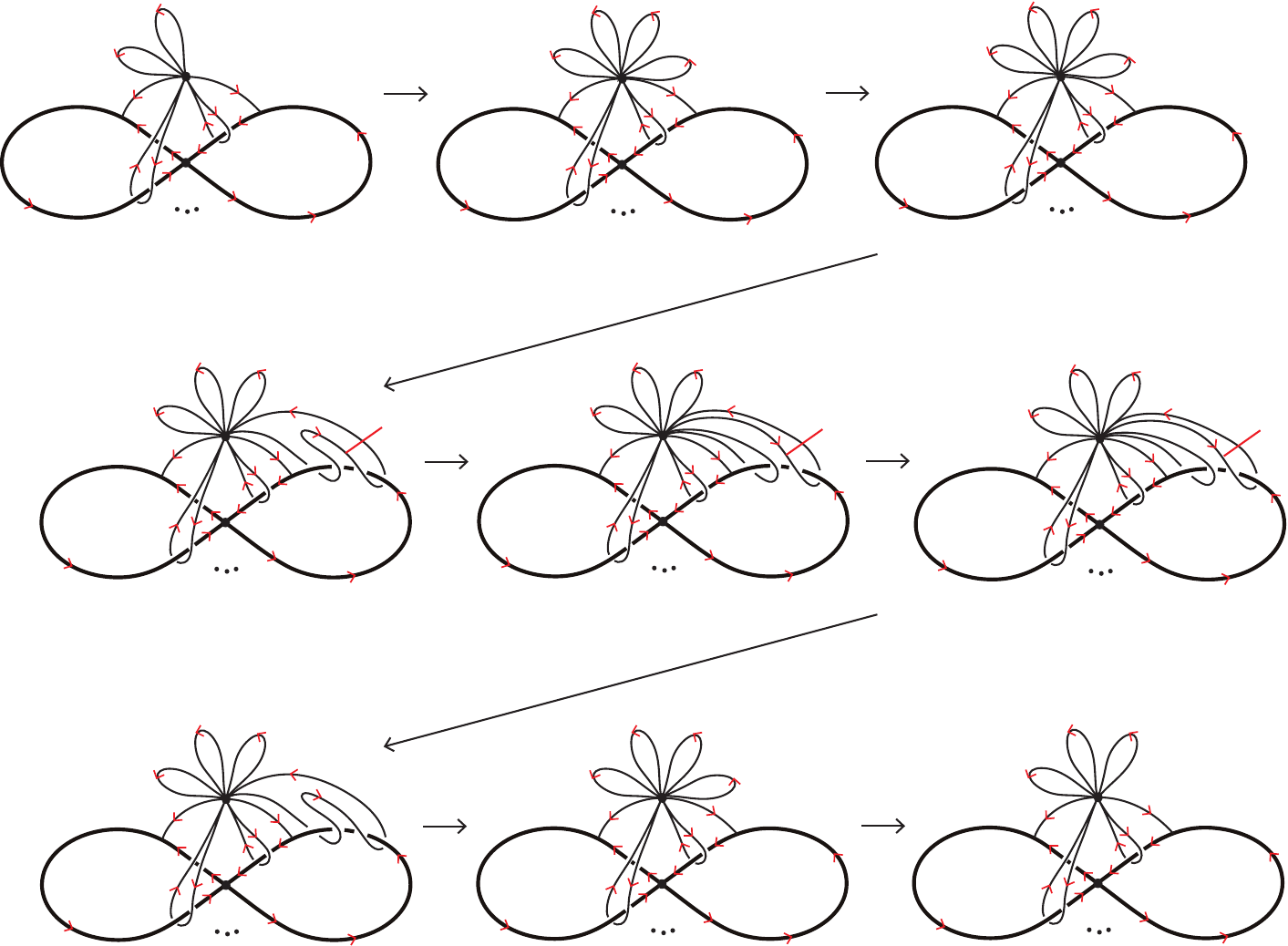}

\begin{picture}(400,0)(0,0)
\put(203,293){\color{red} $1$}
\put(217,273){\color{red} $1$}
\put(334,293){\color{red} $1$}
\put(348,273){\color{red} $c$}
\put(120,253){\color{red} $d$}
\put(250,253){\color{red} $d$}
\put(380,253){\color{red} $d$}

\put(84,186){\color{red} $1$}
\put(124,165){\color{red} $d^{-1} c d$}
\put(94,175){\color{red} $c$}
\put(255,165){\color{red} $d^{-1} c d$}
\put(214,186){\color{red} $1$}
\put(345,186){\color{red} $d^{-1} c d$}
\put(225,175){\color{red} $c$}
\put(386,165){\color{red} $d^{-1} c d$}
\put(355,175){\color{red} $c$}
\put(132,145){\color{red} $d$}
\put(262,145){\color{red} $d$}
\put(392,145){\color{red} $d$}

\put(84,79){\color{red} $d^{-1} c d$}
\put(94,67){\color{red} $c$}
\put(214,79){\color{red} $d^{-1} c d$}
\put(230,60){\color{red} $c$}
\put(345,79){\color{red} $d^{-1} c d$}
\put(132,36){\color{red} $d$}
\put(262,36){\color{red} $d$}

\end{picture}
\caption{A sequence of moves creating a loop colored by $d^{-1}cd$, where $c\in \{\col(\alpha)\mid \alpha\in \A_{D_{\Gamma'}}\}$ and $d\in H=\ang{\col(\alpha)\mid \alpha\in \A_{D_B\cup D_{\Gamma'}}}$.
Here, the second and eighth (resp.\ fifth) moves follow from Figure~16 (resp.\ Figure~13) in \cite{NKTK24}.}
\label{fig:conjugate}
\end{figure}
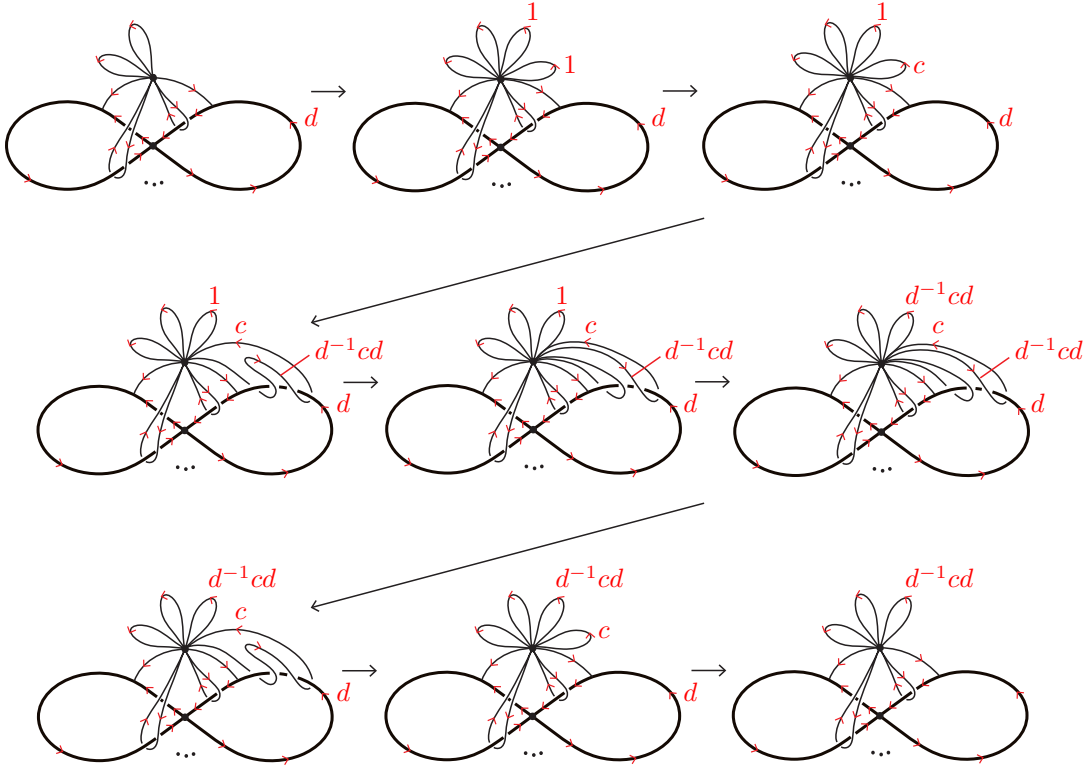

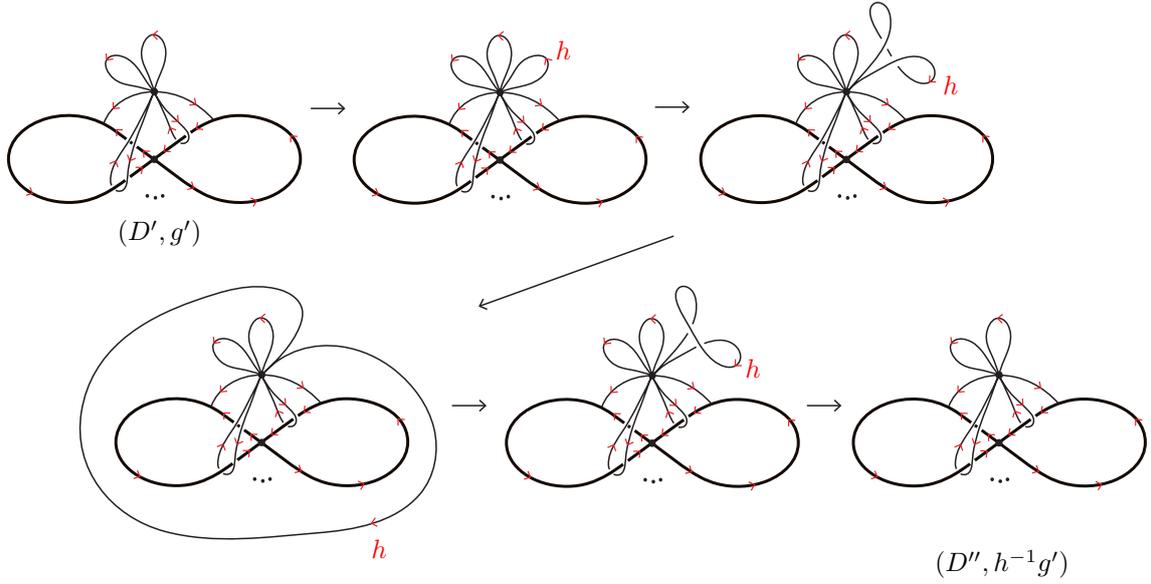
\begin{figure}[h]
\centering\includegraphics[width=1\textwidth]{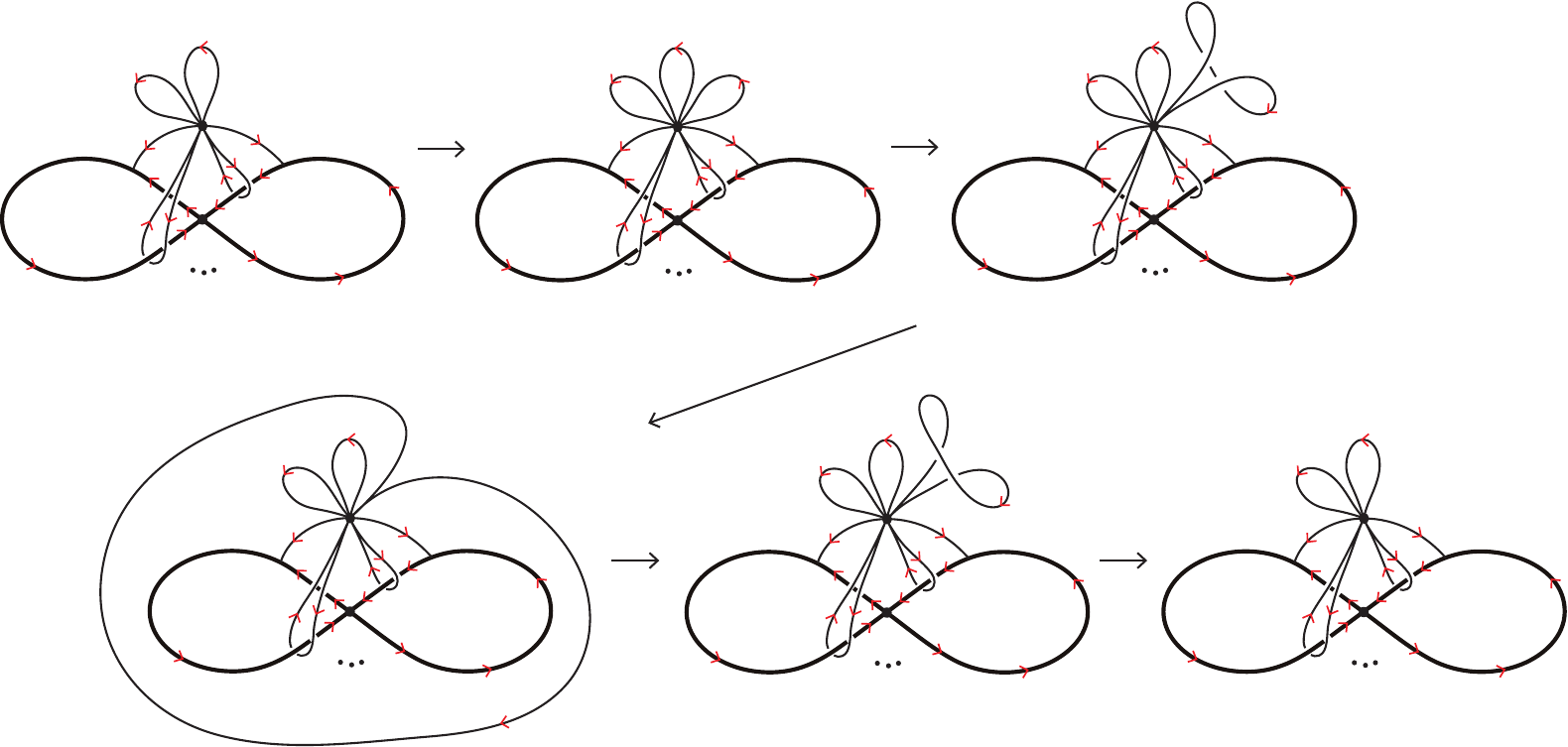}
\begin{picture}(400,0)(0,0)
\put(28,125){$(D', g')$}
\put(192,193){\color{red} $h$}
\put(337,180){\color{red} $h$}
\put(123,5){\color{red} $h$}
\put(263,73){\color{red} $h$}
\put(334,0){$(D'', h^{-1}g')$}
\end{picture}
\caption{A sequence of moves changing $g'$ to $h^{-1}g'$ in Definition~\ref{def:colored_diagram}(3).}
\label{fig:returning}
\end{figure}

%%%%
\appendix

%%%
\section{Classification example: a system on the solid torus with order parameter space $S^3 / Q$ and no boundary defects}%\texorpdfstring{}
\label{sec: Classification example}
Let $X_G = S^3 /Q$, $k = 1$ and $P = \emptyset$, that is, we are considering 
a system defined on the solid torus $M=D^2 \times S^1$ that has 
$S^3 / Q$ as its order parameter space and no boundary defects. 
Then $\pi_1 (\partial M , \ast_\partial )$ is the free abelian group of 
rank two generated by the loops $\alpha = \alpha_1$ and $\beta = \beta_1$.  
Let $f_0\colon \partial M \to S^3 / Q$ satisfy $(f_0)_* (\alpha) = i \in Q$. 
We can classify the equivalence classes of global defect configurations with 
the boundary condition $f_0$ as follows. 

Suppose first that $(f_0)_* (\beta) = 1$. 
Then, we have 
$\SS_{G, \Im (f_0)_*, N_{\beta \gamma}} = \SS_{Q, \langle i \rangle, \{1\}}$. 
As we have seen in Example \ref{example: triple}, 
this set consists of nine elements, and from this, 
we can check easily that $\SS_{Q, \langle i \rangle, \{1\}} /\conj$ consists of six elements (see \Cref{cor:comp_classif_conjugacy}). 
The element $(\langle i \rangle, \{ 1 \}, [1] )$ corresponds to the empty defect, 
and we can choose a representative of the equivalence class of colored diagrams corresponding to each of the other five elements as shown in 
Figure~\ref{fig:example_classification1}. 
\begin{figure}[htbp]
\centering\includegraphics[width=0.95\textwidth]{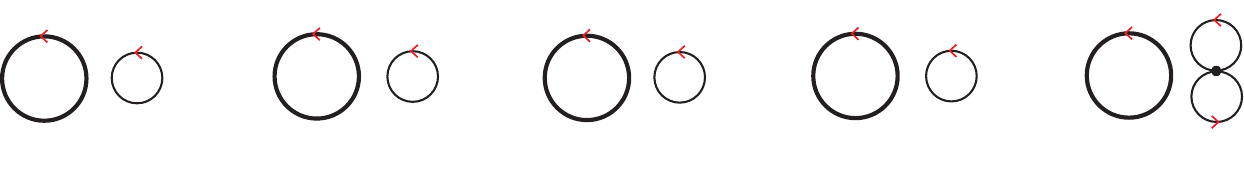}
\begin{picture}(400,0)(0,0)

\put(40,59){\footnotesize \color{red}$-1$}
\put(130,59){\footnotesize \color{red}$i$}
\put(218,60){\footnotesize \color{red}$j$}
\put(306,59){\footnotesize \color{red}$k$}
\put(392,71){\footnotesize \color{red}$j$}
\put(392,22){\footnotesize \color{red}$k$}

\put(10,65){\footnotesize \color{red}$i$}
\put(100,65){\footnotesize \color{red}$i$}
\put(188,65){\footnotesize \color{red}$i$}
\put(275,65){\footnotesize \color{red}$i$}
\put(364,65){\footnotesize \color{red}$i$}

\put(-4,0){$(\langle  i \rangle, \langle  -1 \rangle, [1])$}
\put(91,0){$(\langle  i \rangle, \langle  i \rangle, [1])$}
\put(179,0){$(Q , \langle j \rangle, [1])$}
\put(263,0){$(Q, \langle  k \rangle , [1])$}
\put(355,0){$(Q,Q,[1])$}

\put(26,18){$\updownarrow$}
\put(112,18){$\updownarrow$}
\put(197,18){$\updownarrow$}
\put(284,18){$\updownarrow$}
\put(373,18){$\updownarrow$}

\end{picture}
\caption{Representatives of 
the equivalence classes of colored diagrams corresponding to non-trivial elements of 
$\SS_{Q, \langle i \rangle, \{1\}} /\conj$ in the case where 
$(f_0)_* (\beta) = 1$.}
\label{fig:example_classification1}
\end{figure}
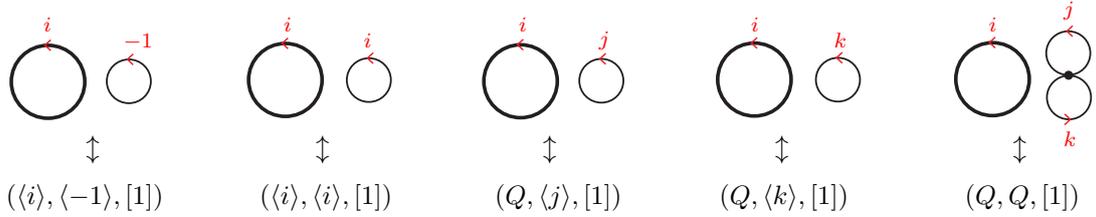

Next, suppose that $(f_0)_* (\beta) = -1$. 
Then, we have 
$\SS_{G, \Im (f_0)_*, N_{\beta \gamma}} = \SS_{Q, \langle i \rangle, \langle -1 \rangle}$. 
As we have seen in Example~\ref{example: triple}, 
this set consists of seven elements, and from this, 
we see that $\SS_{Q, \langle i \rangle, \langle -1 \rangle} /\conj$ consists of five elements, where 
a representative of the equivalence class of colored diagrams corresponding to each of those elements can be seen in Figure~\ref{fig:example_classification2}.

\begin{figure}[htbp]
\centering\includegraphics[width=0.95\textwidth]{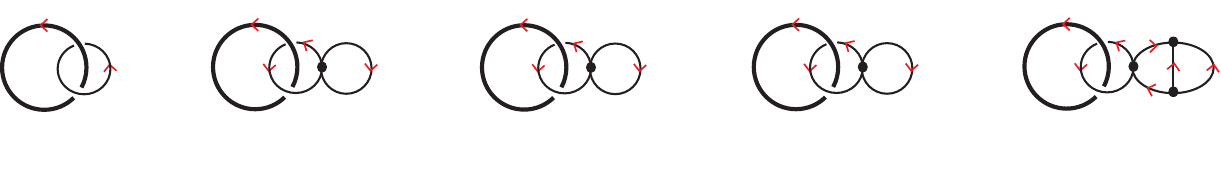}
\begin{picture}(500,0)(0,0)

\put(23,65){\footnotesize \color{red}$i$}
\put(51,44){\footnotesize \color{red}$-1$}

\put(86,44){\footnotesize \color{red}$-1$}
\put(138,44){\footnotesize \color{red}$i$}
\put(108,58){\footnotesize \color{red}$-1$}
\put(93,65){\footnotesize \color{red}$i$}

\put(175,44){\footnotesize \color{red}$-1$}
\put(227,45){\footnotesize \color{red}$j$}
\put(197,58){\footnotesize \color{red}$-1$}
\put(186,65){\footnotesize \color{red}$i$}

\put(266,44){\footnotesize \color{red}$-1$}
\put(317,45){\footnotesize \color{red}$k$}
\put(287,58){\footnotesize \color{red}$-1$}
\put(275,65){\footnotesize \color{red}$i$}

\put(364,65){\footnotesize \color{red}$i$}
\put(376,58){\footnotesize \color{red}$-1$}
\put(356,44){\footnotesize \color{red}$-1$}
\put(391,58){\footnotesize \color{red}$i$}
\put(391,30){\footnotesize \color{red}$i$}
\put(403,45){\footnotesize \color{red}$j$}
\put(416,44){\footnotesize \color{red}$k$}

\put(1,0){$(\langle  i \rangle, \langle  -1 \rangle, [1])$}
\put(83,0){$(\langle  i \rangle , \langle i \rangle, [1])$}
\put(175,0){$(Q, \langle j \rangle , [1])$}
\put(265,0){$(Q,\langle k \rangle , [1])$}
\put(360,0){$(Q, Q , [1])$}

\put(26,18){$\updownarrow$}
\put(105,18){$\updownarrow$}
\put(195,18){$\updownarrow$}
\put(285,18){$\updownarrow$}
\put(378,18){$\updownarrow$}

\end{picture}
\caption{Representatives of 
the equivalence classes of colored diagrams corresponding to elements of 
$\SS_{Q, \langle i \rangle, \langle -1 \rangle } /\conj$ in the case where 
$(f_0)_* (\beta) = -1$.}
\label{fig:example_classification2}
\end{figure}
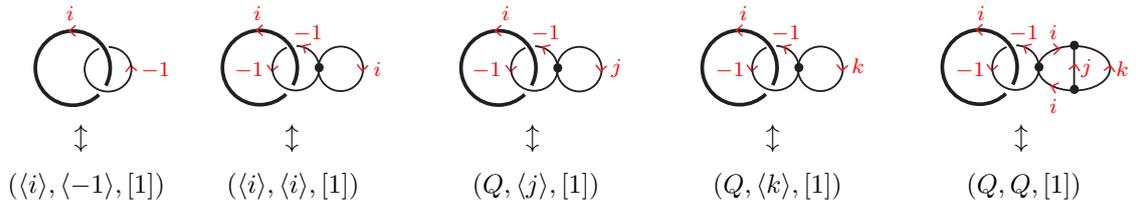

Finally, suppose that $(f_0)_* (\beta) = \pm i$. 
In this case, we have 
$\SS_{G, \Im (f_0)_*, N_{\beta \gamma}} = \SS_{Q, \langle i \rangle, \langle i \rangle}$. 
As we have seen in Example \ref{example: triple}, 
this set consists of three elements, and from this, 
we see that $\SS_{Q, \langle i \rangle, \{1\}} /\conj$ consists of 
two elements. 
A representative of the equivalence class of colored diagrams corresponding to each of those elements for $(f_0)_* (\beta) = \pm i$ can be seen in Figure~\ref{fig:example_classification3}.  
\begin{figure}[htbp]
\centering\includegraphics[width=0.80\textwidth]{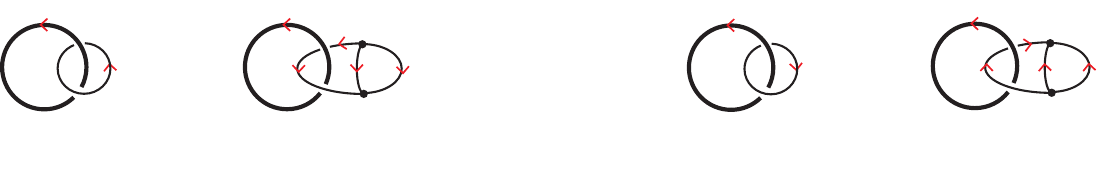}
\begin{picture}(400,0)(0,0)

\put(42,67){\footnotesize \color{red}$i$}
\put(67,47){\footnotesize \color{red}$i$}

\put(118,67){\footnotesize \color{red}$i$}
\put(116,48){\footnotesize \color{red}$i$}
\put(143,47){\footnotesize \color{red}$k$}
\put(158,48){\footnotesize \color{red}$j$}
\put(134,60){\footnotesize \color{red}$i$}

\put(255,67){\footnotesize \color{red}$i$}
\put(280,47){\footnotesize \color{red}$i$}

\put(331,67){\footnotesize \color{red}$i$}
\put(330,48){\footnotesize \color{red}$i$}
\put(357,48){\footnotesize \color{red}$j$}
\put(371,47){\footnotesize \color{red}$k$}
\put(347,60){\footnotesize \color{red}$i$}

\put(30,0){$(\langle  i \rangle, \langle  i \rangle, [1])$}
\put(113,0){$(Q, Q , [1])$}

\put(50,18){$\updownarrow$}
\put(129,18){$\updownarrow$}

\put(238,0){$(\langle  i \rangle, \langle  i \rangle, [1])$}
\put(322,0){$(Q, Q , [1])$}

\put(256,18){$\updownarrow$}
\put(338,18){$\updownarrow$}

\end{picture}
\caption{Representatives of 
the equivalence classes of colored diagrams corresponding to elements of 
$\SS_{Q, \langle i \rangle, \langle i \rangle } /\conj$ in the case where 
$(f_0)_* (\beta) = i$ (left) or $(f_0)_* (\beta) = -i$ (right).}
\label{fig:example_classification3}
\end{figure}
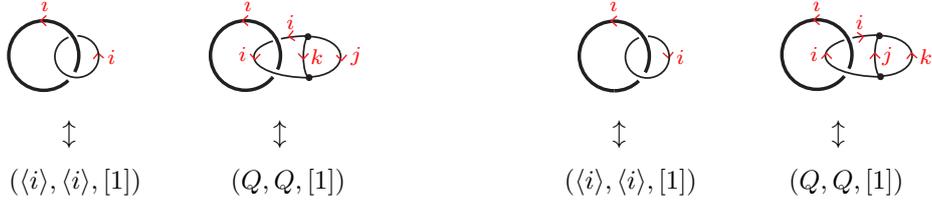

Note that the list of representatives of the colored diagrams obtained above is exactly the one we introduced in Section~\ref{sec:intro} with Figure~\ref{fig:example_classification_intro}. 
By Corollary~\ref{cor:comp_classif_conjugacy}, 
this gives a complete classification, 
up to free homotopy, of non-trivial defect configurations of a system defined on the solid torus $M$ with order parameter space $S^3 / Q$ and no boundary defects, subject to the boundary condition 
$f_0\colon \partial M \to S^3 / Q$ with $(f_0)_* (\alpha) = i$.

Obviously, there are infinitely many colored diagrams 
in each equivalence class of colored diagrams. 
For example, the equivalence class of colored diagrams 
corresponding to the element $(\langle i \rangle, \langle i \rangle, [1])$ of 
$\SS_{Q, \langle i \rangle , \langle -1\rangle} /\conj$ contains the colored diagrams depicted in Figure~\ref{fig:example_equivalence1}. 
\begin{figure}[htbp]
\centering\includegraphics[width=0.7\textwidth]{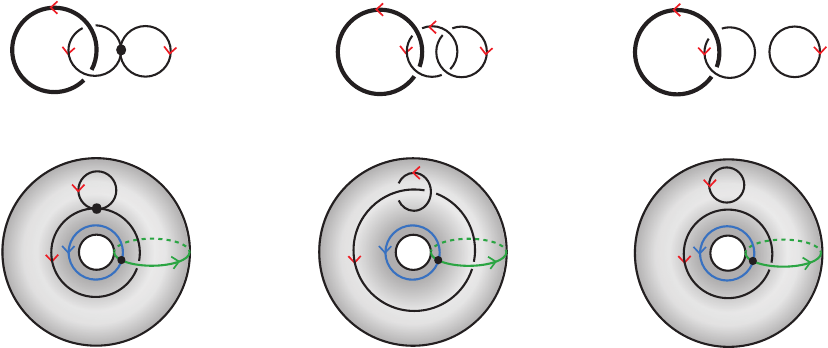}
\begin{picture}(400,0)(0,0)

\put(140,117){$\sim$}
\put(254,117){$\sim$}

\put(62,117){\footnotesize \color{red}$-1$}
\put(116,117){\footnotesize \color{red}$i$}
\put(85,131){\footnotesize \color{red}$-1$}
\put(183,117){\footnotesize \color{red}$-1$}
\put(230,117){\footnotesize \color{red}$i$}
\put(203,131){\footnotesize \color{red}$-1$}
\put(291,117){\footnotesize \color{red}$-1$}
\put(351,117){\footnotesize \color{red}$i$}

\put(69,139){\footnotesize \color{red}$i$}
\put(187,139){\footnotesize \color{red}$i$}
\put(294,139){\footnotesize \color{red}$i$}

\put(84,91){$\updownarrow$}
\put(196,91){$\updownarrow$}
\put(310,91){$\updownarrow$}

\put(73,68){\footnotesize {\color{red} $i$}}
\put(205,74){\footnotesize {\color{red} $i$}}
\put(301,68){\footnotesize {\color{red} $i$}}

\put(56,42){\footnotesize {\color{red} $-1$}}
\put(166,42){\footnotesize {\color{red} $-1$}}
\put(284,42){\footnotesize {\color{red} $-1$}}

\put(81,45){\footnotesize {\color{blue} $\alpha$}}
\put(196,45){\footnotesize {\color{blue} $\alpha$}}
\put(309,45){\footnotesize {\color{blue} $\alpha$}}

\put(122,44){\footnotesize {\color{teal} $\beta$}}
\put(237,44){\footnotesize {\color{teal} $\beta$}}
\put(350,44){\footnotesize {\color{teal} $\beta$}}

\end{picture}
\caption{Some equivalent colored diagrams corresponding to the element $(\langle i \rangle, \langle i \rangle, [1])$ of $\SS_{Q, \langle i \rangle , \langle -1\rangle} /\conj$.}
\label{fig:example_equivalence1}
\end{figure}
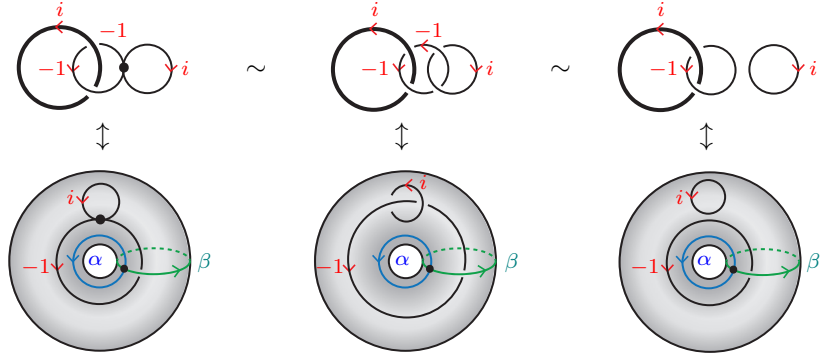
Similarly, 
the three colored diagrams depicted in Figure \ref{fig:example_equivalence2} 
belong to the same equivalence class because they all correspond to 
the same element $(Q, Q, [1])$ of 
$\SS_{Q, \langle i \rangle , \langle i\rangle} /\conj$. 
\begin{figure}[htbp]
\centering\includegraphics[width=0.70\textwidth]{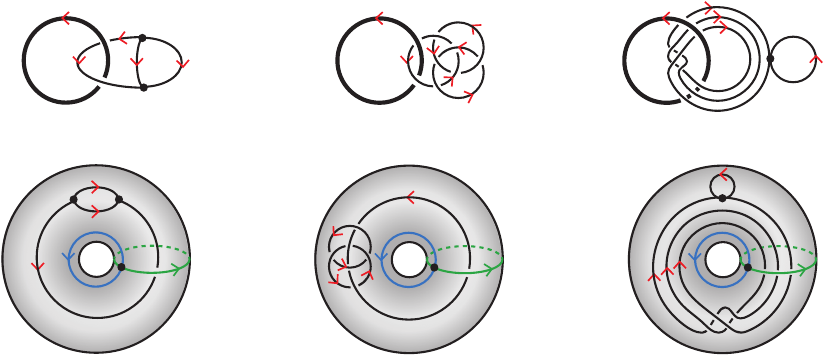}
\begin{picture}(400,0)(0,0)
\put(142,117){$\sim$}
\put(252,117){$\sim$}

\put(73,117){\footnotesize \color{red}$i$}
\put(103,117){\footnotesize \color{red}$k$}
\put(120,117){\footnotesize \color{red}$j$}
\put(192,117){\footnotesize \color{red}$i$}
\put(226,132){\footnotesize \color{red}$j$}
\put(225,101){\footnotesize \color{red}$k$}
\put(288,117){\footnotesize \color{red}$i$}
\put(350,117){\footnotesize \color{red}$j$}

\put(73,138){\footnotesize \color{red}$i$}
\put(188,138){\footnotesize \color{red}$i$}
\put(293,138){\footnotesize \color{red}$i$}

\put(84,90){$\updownarrow$}
\put(196,90){$\updownarrow$}
\put(310,90){$\updownarrow$}

\put(89,76){\footnotesize {\color{red} $j$}}
\put(78,58){\footnotesize {\color{red} $k$}}
\put(173,62){\footnotesize {\color{red} $j$}}
\put(172,30){\footnotesize {\color{red} $k$}}
\put(320,72){\footnotesize {\color{red} $j$}}

\put(57,42){\footnotesize {\color{red} $i$}}
\put(198,72){\footnotesize {\color{red} $i$}}
\put(281,38){\footnotesize {\color{red} $i$}}

\put(81,45){\footnotesize {\color{blue} $\alpha$}}
\put(196,45){\footnotesize {\color{blue} $\alpha$}}
\put(309,45){\footnotesize {\color{blue} $\alpha$}}

\put(122,44){\footnotesize {\color{teal} $\beta$}}
\put(237,44){\footnotesize {\color{teal} $\beta$}}
\put(350,44){\footnotesize {\color{teal} $\beta$}}
\end{picture}
\caption{Some equivalent colored diagrams corresponding to the element $(Q, Q, [1])$ of 
$\SS_{Q, \langle i \rangle , \langle i\rangle} /\conj$.}
\label{fig:example_equivalence2}
\end{figure}
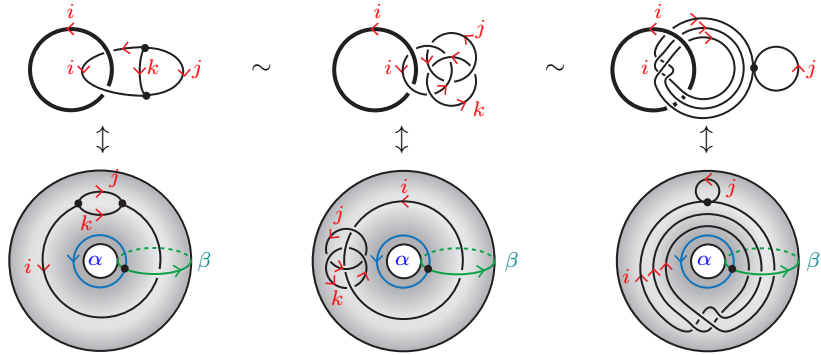
Using Corollary~\ref{cor:comp_classif_conjugacy}, we can determine whether two diagrams of defect configurations are equivalent by simply comparing the corresponding elements of $\SS_{G, \Im (f_0)_*, N_{\beta \gamma}} / \conj$, which are easily computed from the diagrams.

%%%
\section{The subgroups of the binary octahedral group $\BOct$}%\texorpdfstring{}
\label{sec: The subgroups of the binary octahedral group}

The following is the list of all subgroups of $\BOct$. 

\begin{itemize}
\item Subgroups of $Q$. 
\begin{enumerate}[label=(\roman*)]
\item $\{ 1 \} = \langle 1 \rangle$. 
\item $\langle -1 \rangle \cong \Z / 2 \Z$. 
\item $\langle i \rangle  \cong \Z / 4 \Z$, $\langle j \rangle$, $\langle k \rangle$.
\item \textbf{(order $8$)} $Q = \langle i,j \rangle$. 
\end{enumerate}
\item Subgroups of $\BTet$ not contained in $Q$. 
\begin{enumerate}[label=(\roman*),resume]
\item $\langle c^2 \rangle = \{ 1, c^2, c^^4 = -c \} \cong \Z / 3 \Z$, 
$\langle \alpha^2 \rangle$, $\langle \beta^2 \rangle$, $\langle \gamma^2 \rangle$. 
\item $\langle c \rangle = \{ \pm 1, \pm c, \pm c^2 \}  \cong \Z / 6 \Z$, 
$\langle \alpha \rangle$, 
$\langle \beta \rangle$, $\langle \gamma \rangle$. 
\item 
$\BTet= \langle i , c \rangle $. 
\end{enumerate}
\item Subgroups of $\BOct$ not contained in $\BTet$.
\begin{enumerate}[label=(\roman*),resume]
\item 
$\left\langle \frac{1}{\sqrt{2}} (i + j) \right\rangle = 
\left\{ \pm 1, \pm \frac{1}{\sqrt{2}} (i+j) \right\} \cong \Z / 4 \Z$, 
$\left\langle \frac{1}{\sqrt{2}} (j + k) \right\rangle$, 
$\left\langle \frac{1}{\sqrt{2}} (k + i) \right\rangle$, 
$\left\langle \frac{1}{\sqrt{2}} (i - j) \right\rangle  = 
\left\{ \pm 1, \pm \frac{1}{\sqrt{2}} (i-j) \right\}$, 
$\left\langle \frac{1}{\sqrt{2}} (j - k) \right\rangle$, 
$\left\langle \frac{1}{\sqrt{2}} (k - i) \right\rangle$. 
\item 
$\left\langle \frac{1}{\sqrt{2}} (1 + i) \right\rangle = 
\left\{ \pm 1, \pm i, \frac{1}{\sqrt{2}} (\pm 1, \pm i) \right\}$, $\left\langle \frac{1}{\sqrt{2}} (1 + j) \right\rangle  \cong \Z / 8 \Z$, 
$\left\langle \frac{1}{\sqrt{2}} (1 + k) \right\rangle$. 
\item  
$Q_8(i) \defeq \left\langle i, \frac{1}{\sqrt{2}} (j+k) \right\rangle = 
\left\{ \pm 1, \pm i , \frac{1}{\sqrt{2}} (j+k) , \pm \frac{1}{\sqrt{2}} (j-k) \right\} \cong Q \cong D_2^*$, 
$Q_8(j) \defeq \left\langle j, \frac{1}{\sqrt{2}} (k+i) \right\rangle$, 
$Q_8(k) \defeq \left\langle k, \frac{1}{\sqrt{2}} (i+j) \right\rangle$.
\item 
$Q_{12}(c) \defeq \left\langle c, \frac{1}{\sqrt{2}} (i-j) \right\rangle = 
\left\{ \pm 1, \pm c , \pm c^2, 
\pm \frac{1}{\sqrt{2}} (i-j) , \pm \frac{1}{\sqrt{2}} (j-k), 
\pm \frac{1}{\sqrt{2}} (k-i) \right\} \cong D_3^*$, 
$Q_{12} (\alpha) \defeq \left\langle \alpha, \frac{1}{\sqrt{2}} (i+j) \right\rangle = 
\left\{ \pm 1, \pm \alpha , \pm \alpha^2, 
\pm \frac{1}{\sqrt{2}} (i+j) , \pm \frac{1}{\sqrt{2}} (j-k), \pm \frac{1}{\sqrt{2}} (k+i) \right\}$, 
$Q_{12} (\beta)$, $Q_{12} (\beta)$.
\item 
$Q_{16}(i) \defeq \left\langle i, \frac{1}{\sqrt{2}} (1+k) \right\rangle = 
Q \cup \left\{ \frac{1}{\sqrt{2}} ( \pm i \pm j), 
 \pm \frac{1}{\sqrt{2}} ( \pm 1 \pm k) 
\right\} \cong D_4^*$, 
$Q_{16}(j)$, $Q_{16}(k)$.
\item  
$\BOct = \left\langle c, \frac{1}{\sqrt{2}} (i+j) \right\rangle$. 
\end{enumerate}
\end{itemize}
The subgroups within each item of the above list are conjugate to each other. 
Thus, only $\{1\}$, $\langle -1 \rangle$, $Q$, $\BTet$ and $\BOct$ are normal subgroups of $\BOct$, 
and there are totally thirteen conjugacy classes of subgroups.

%\bibliographystyle{abbrv}%alpha
%\bibliography{Defect-with-boundary}
\printbibliography

\end{document}